\documentclass[twocolumn]{bmcart}
\usepackage[utf8]{inputenc}
\usepackage{amsthm,amsmath,amssymb}
\usepackage{leftidx}
\usepackage{algorithm,algpseudocode}
\usepackage{enumitem}
\usepackage{graphicx}
\usepackage{url}
\usepackage{xcolor}
\usepackage{epstopdf}
\usepackage{xparse}
\usepackage{nameref}
\usepackage{pgfplots}
\usepackage{filecontents}
\usepackage{threeparttable, tablefootnote}
\usepackage{hyperref}

\pgfplotsset{
    legend entry/.initial=,
    every axis plot post/.code={%
        \pgfkeysgetvalue{/pgfplots/legend entry}\tempValue
        \ifx\tempValue\empty
            \pgfkeysalso{/pgfplots/forget plot}%
        \else
            \expandafter\addlegendentry\expandafter{\tempValue}%
        \fi
    },
}

\newcommand\norm[1]{\lvert\lvert#1\rvert\rvert}
\newcommand\epi{\bar{\pi}}
\newcommand\esigma{\bar{\sigma}}
\newcommand\eiota{\bar{\iota}}
\newcommand\spi{\eiota {\epi}^{-1}}

\def\doubleunderline#1{\underline{\underline{#1}}}
\newcommand{\floor}[1]{\left\lfloor #1 \right\rfloor}

\newtheorem{theorem}{Theorem}
\newtheorem{corollary}[theorem]{Corollary}
\newtheorem{lemma}[theorem]{Lemma}
\newtheorem{example}[theorem]{Example}
\newtheorem{proposition}[theorem]{Proposition}

\setlist[enumerate,1]{label=(\arabic*),ref=(\arabic*)}
\setlist[enumerate,2]{label=(\alph*),ref=(\arabic{enumi}-\alph*)}
\setlist[enumerate,3]{label=(\roman*),ref=(\arabic{enumi}-\alph{enumii}-\roman*)}

\ExplSyntaxOn
\NewDocumentCommand{\cycle}{ O{\;} m }
{
	(
	\alec_cycle:nn { #1 } { #2 }
	)
}

\seq_new:N \l_alec_cycle_seq
\cs_new_protected:Npn \alec_cycle:nn #1 #2
{
	\seq_set_split:Nnn \l_alec_cycle_seq { , } { #2 }
	\seq_use:Nn \l_alec_cycle_seq { #1 }
}
\ExplSyntaxOff

\ExplSyntaxOn
\NewDocumentCommand{\cyclen}{ O{} m }
{
	(
	\alec_cyclen:nn { #1 } { #2 }
	)
}

\seq_new:N \l_alec_cyclen_seq
\cs_new_protected:Npn \alec_cyclen:nn #1 #2
{
	\seq_set_split:Nnn \l_alec_cyclen_seq { , } { #2 }
	\seq_use:Nn \l_alec_cyclen_seq { #1 }
}
\ExplSyntaxOff

\begin{document}
	\begin{frontmatter}
		\begin{fmbox}
            \dochead{Research}
            
            \title{A new $1.375$-approximation algorithm for Sorting By Transpositions}
            
            \author[
            	addressref={aff1},             
            	corref={aff1},                 
            	email={laugustogarcia@gmail.com} 
            ]{\inits{LAGS}\fnm{Luiz Augusto G.} \snm{Silva}}
            \author[
            	addressref={aff3},     
            	email={luis@ic.uff.br}
            ]{\inits{LABK}\fnm{Luis Antonio B.} \snm{Kowada}}
            \author[
            	addressref={aff2},     
            	email={norai@unb.br}  
            ]{\inits{NRR}\fnm{Noraí Romeu} \snm{Rocco}}
            \author[
            	addressref={aff1},       
            	email={mariaemilia@unb.br}  
            ]{\inits{MEMTW}\fnm{Maria Emília M. T.} \snm{Walter}}

            \address[id=aff1]{
              \orgname{Departmento de Ciência da Computação, Universidade de Brasília}, 
              \city{Brasília},              
              \cny{Brazil}                  
            }
            \address[id=aff2]{
              \orgname{Departmento de Matemática, Universidade de Brasília},
              \city{Brasília},
              \cny{Brazil}
            }
            \address[id=aff3]{
              \orgname{Instituto de Computação, Universidade Federal Fluminense},
              \city{Niterói},
              \cny{Brazil}
            }

		\begin{abstractbox}
    		\begin{abstract}
    		\parttitle{Background}
			In genome rearrangements, the mutational event \emph{transposition} swaps two adjacent blocks of genes in one chromosome. The \emph{Transposition Distance Problem} (TDP) aims to find the minimum number of transpositions required to transform one chromosome into another (transposition distance), both represented as permutations. The pair of permutations can be transformed into another pair with the same distance where the target permutation is the identity, making TDP equivalent to the problem of \emph{Sorting by Transpositions} (SBT).

			In 2012, SBT was proven to be $\mathcal{NP}$-hard and the best approximation algorithm with a $1.375$ ratio was proposed in 2006 by Elias and Hartman. Their algorithm employs \emph{simplification}, a technique used to transform an input permutation $\pi$ into a \emph{simple permutation} $\hat{\pi}$, presumably easier to handle with. The permutation $\hat{\pi}$ is obtained by inserting new symbols into $\pi$ in a way that the lower bound of the transposition distance of $\pi$ is kept on $\hat{\pi}$. The simplification is guaranteed to keep the lower bound, not the transposition distance. A sequence of operations sorting $\hat{\pi}$ can be mimicked to sort $\pi$.

    		\parttitle{Results and conclusions}
			First, we show that the algorithm of Elias and Hartman (EH algorithm) may require one extra transposition above the approximation ratio of $1.375$, depending on how the input permutation is simplified. Next, using an algebraic approach, we propose a new upper bound for the transposition distance and a new $1.375$-approximation algorithm to solve SBT skipping simplification and ensuring the approximation ratio of $1.375$ for all the permutations in the Symmetric Group $S_n$.

			We implemented our algorithm and EH's. Regarding the implementation of the EH algorithm, two issues needed to be fixed. We tested both algorithms against all permutations of size $n$, $2\leq n \leq 12$. The results show that the EH algorithm exceeds the approximation ratio of $1.375$ for permutations with a size greater than $7$. Overall, the average of the distances computed by our algorithm is a little better than the average of the ones computed by the EH algorithm and the execution times are similar. The percentage of computed distances that are equal to transposition distance, computed by both algorithms are also compared with others available in the literature. Finally, we investigate the performance of both implementations on longer permutations of maximum length $500$. From this experiment, we conclude that both maximum and average distances computed by our algorithm are a little better than the ones computed by the EH algorithm. We also conclude that the running times of both algorithms are similar.

    		\end{abstract}

		\end{abstractbox}
		\end{fmbox}
	\end{frontmatter}
	
	\section*{Background}
	
	It is known from previous research that the genomes of different species may present essentially the same set of genes in their DNA strands, although not in the same order~\cite{NadeauTaylor1984, PalmerHerbon1988}, suggesting the occurrence of mutational events that affect large portions of DNA. These are presumably rare events and, therefore, may provide important clues for the reconstruction of the evolutionary history among species~\cite{koonin2005orthologs,YuePhylogenetics}. One such event is the \emph {transposition}, which swaps the position of two adjacent blocks of genes in one chromosome. Considering that there are no duplicated genes, each gene can be represented by an integer and the chromosome by a permutation, then the \textsc{Transposition Distance Problem} (TDP) aims to find the minimum number of transpositions required to transform one chromosome into another. The pair of permutations can be transformed into another pair with the same distance where the target permutation is the identity, making TDP is equivalent to the problem of \textsc{Sorting by Transpositions} (SBT).
	
	The first approximation algorithm to solve SBT was devised in 1998 by Bafna and Pevzer~\cite{BafnaPevzner1998}, with a $1.5$ ratio, based on the properties of a structure called the cycle graph. In 2006, Elias and Hartman~\cite{EliasHartman2006} presented a $1.375$-approximation algorithm (EH algorithm) with time complexity $O(n^2)$, the best known approximation solution so far for SBT, also based on the cycle graph. In 2012, Bulteau, Fertin and Rusu~\cite{bulteau2012sorting} demonstrated that SBT is $\mathcal{NP}$-hard.
	
	In a later study, the time complexity of the EH algorithm was improved to $O(n\log n)$ by Cunha et al.~\cite{Cunha2015}. Improvements to the EH algorithm, including heuristics, were proposed by Dias and Dias~\cite{dias2010improved, DiasDias2013}.
	
	Other studies, using different approaches, other than the cycle graph, were also published. For instance, Hausen et al.~\cite{Hausen2010} studied SBT using a structure named toric graph, which was previously devised by Erikson et al.~\cite{eriksson2001sorting}, used by the later ones to derive the upper bound of $\floor{\frac{2n - 2}{3}}$ for the transposition diameter, the best known so far for SBT. Galv\~ao and Dias~\cite{galvao2012approximation} studied solutions for SBT using three different structures: permutation codes, a concept previously introduced by Benoît-Gagné and Hamel~\cite{benoit2007new}; breakpoint diagram\footnote{Do not confuse with breakpoint graph.}, introduced by Walter et al.~\cite{Walter:2000:NAA:829519.830850}; and longest increasing subsequence, introduced by Guyer et al.~\cite{guyer1997subsequence}. Rusu~\cite{rusu2017log}, on the other hand, used a structure called log-list, formerly devised with the name link-cut trees by Sleator and Tarjan~\cite{sleator1983data}, to derive another $O(n\log n)$ algorithm for SBT. In addition to these, recently, other studies have been proposed involving variations of the transposition event. As examples, Lintzmayer et al.~\cite{lintzmayer2017sorting} studied the problem of \textsc{Sorting by Prefix and Sufix Transpositions}, as well as other problems combining variations of the transposition event with variations of the reversal event. Oliveira et al.~\cite{oliveira20203} studied the transposition distance between two genomes considering intergenic regions, a problem they called \textsc{Sorting Permutations by Intergenic Transpositions}.

	Meidanis and Dias~\cite{MeidanisDias2000} and Mira and Meidanis~\cite{mira2005algebraic} were the first authors to propose the use of an algebraic approach to solve TDP, as an alternative to the methods based on the cycle graph. The goal was to provide a more formal approach for solving rearrangement problems using known results from the permutation groups theory. Mira et al.~\cite{Mira2008} have shown the feasibility of using an algebraic approach to solve SBT by formalizing the Bafna and Pevzner's $1.5$-approximation algorithm~\cite{BafnaPevzner1998} using an algebraic tooling.

	This paper is organized as follows. The first result presented are examples of permutations for which the EH algorithm require one extra transposition above the $1.375$ approximation ratio. Then, using an algebraic approach, we propose a new upper bound for the transposition distance. Next, we propose a new algorithm to solve SBT, skipping simplification, ensuring the $1.375$-approximation for all permutations in the Symmetric Group $S_n$. After that, we present experimental results on all permutations of length $n$, $2 \leq n \leq 12$, of implementations of the EH algorithm and ours. The percentage of computed distances that are equal to transposition distance computed by the EH algorithm and ours are compared with others available in the literature. We also investigate the performance of the implementations of both algorithms with longer permutations with size ranging from $20$ to $500$, and compare the results with similar experiments conducted in other studies. Lastly, we report some issues found in the EH algorithm outlined in~\cite{EliasHartman2006} and~\cite{elias20051}, which became evident when we performed the experiments.
	
	\section*{Preliminaries}\label{sec:preliminaries}
	
	Let $\pi=[\pi_1\;\pi_2\dots\pi_n]$ be a permutation. A \emph{transposition} $\rho(i,j,k)$, with $1 \leq i < j < k \leq n + 1$, ``cuts'' the symbols from the interval $[\pi_i,\pi_{j-1}]$ and then ``pastes'' them right after $\pi_{k-1}$. Thus, the application of $\rho(i,j,k)$ on $\pi$, denoted $\rho(i,j,k) \cdot \pi$, yields
	$[\pi_1\dots\pi_{i-1}\;\pi_j\dots\pi_{k-1}\;\pi_i\dots\pi_{j-1}\;\pi_k\dots\pi_n]$, if $k \leq n$; or $[\pi_1\dots\pi_{i-1}\;\pi_j\dots\pi_{k-1}\;\pi_i\dots\pi_{j-1}]$, if $k = n+1$.
	
	Given two permutations $\pi$ and $\sigma$, the \textsc{Transposition Distance Problem} (TDP) corresponds to finding the minimum $t$ (the transposition distance between $\pi$ and $\sigma$) such that the sequence of transpositions $\tau_1$, $\dots$, $\tau_t$ transforms $\pi$ into $\sigma$ i.e., $\tau_t \dots \tau_1 \cdot \pi=\sigma$. Note that the transposition distance between $\pi$ and $\sigma$ equals the transposition distance between $\sigma^{-1} \circ \pi$ and the identity permutation $\iota=[1\;2\dots n]$. The problem of \textsc{Sorting by Transpositions} (SBT) is the problem of finding the transposition distance between a permutation $\pi$ and $\iota$, denoted by $d_t(\pi)$.
	
	\subsection*{Cycle graph}
	\label{sec:cycle-graph}
	
	In the genome rearrangements literature, a widely used graphical representation for a permutation is the cycle graph\footnote{In their work, Elias and Hartman~\cite{EliasHartman2006} use an equivalent circular representation, which they call breakpoint graph.}~\cite{BafnaPevzner1998}. In order to construct the cycle graph of $\pi=[\pi_1 \; \pi_2 \; \dots \; \pi_n]$, we first extend $\pi$ by adding two extra elements $\pi_0=0$ and $\pi_{n+1}=n+1$. Thus, the \emph{cycle graph} of $\pi$, denoted by $G(\pi)$, is a directed graph consisting of a set of vertices $\{+0,\;-1,\;+1,\;-2,\;+2,\;\dots,$ $\; -n,\;+n,\; -(n+1)\}$ and a set of colored (black or gray) edges. For all $1 \leq i \leq n + 1$, the black edges connect $-\pi_{i}$ to $+\pi_{i-1}$. For $0 \leq i \leq n$, the gray edges connect vertex $+i$ to vertex $-(i+1)$. Intuitively, the black edges indicate the current state of the genes, related to their arrangement in the first chromosome represented by $\pi$, while the gray edges indicate the desired order of the genes in the second permutation, represented by $\iota=[1 \; 2 \; \dots \; n]$. In the figures below, the directions of the edges are omitted since they can be easily inferred by observing the signs of the vertices. 
	
	\begin{example}
		Figure~\ref{fig:cycle-graph} shows the cycle graph of $\pi=[4\;3\;2\;1\;8\;7\;6\;5]$ with $9$ black edges, $(-9,+5)$, $(-5,+6)$, $\dots$, $(-3,+4)$, $(-4,+0)$, and $9$ gray edges, $(+0,-1)$, $(+1,-2)$, $(+2,-3)$, $\dots$, $(+7,-8)$, $(+8,-9)$.
	\end{example}	
	
	\begin{figure*}[ht]
		\centering
		\includegraphics[scale=0.65]{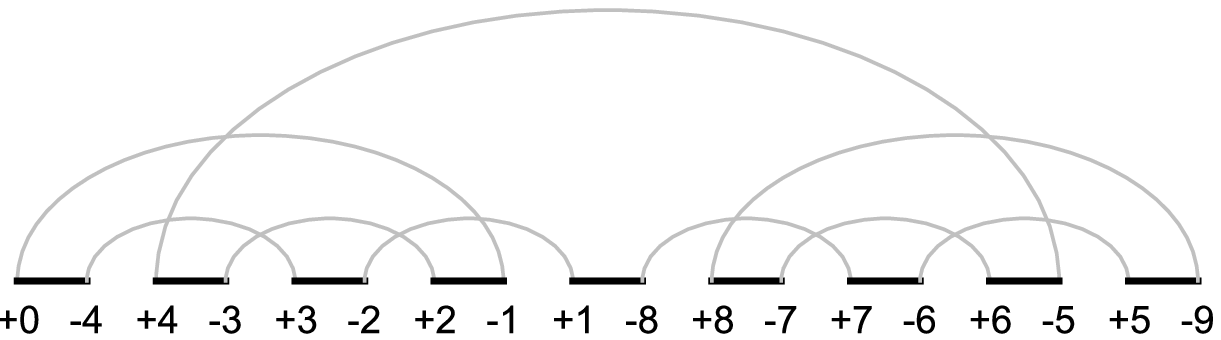}
		\caption{Cycle graph of $[4\;3\;2\;1\;8\;7\;6\;5]$. The black edges are the horizontal ones.}
		\label{fig:cycle-graph}
	\end{figure*}
	
	Both in-degree and out-degree of each vertex in $G(\pi)$ are $1$, corresponding to one black edge entering a vertex $v$ and another gray edge leaving $v$. This induces in $G(\pi)$ a unique decomposition into cycles. 
	A $\kappa$-cycle is a cycle $C$ in $G(\pi)$ with $\kappa$ black edges. In addition, $C$ is said to be a \emph{long cycle}, if $k > 3$, otherwise, $C$ is said to be a \emph {short cycle}. If $\kappa$ is even (odd), then we also say that $C$ is an \emph{even} (\emph{odd}) \emph{cycle}.
	
	The maximum number of $n+1$ cycles in $G(\pi)$ is obtained if and only if $\pi$ is the identity permutation $\iota$. In this case, each cycle is composed of exactly one black edge and one gray edge. Let us denote by $c_{odd}(\pi)$ the number of odd cycles in $G(\pi)$, and $\Delta c_{odd}(\pi,\tau) = c_{odd}(\tau \cdot \pi) - c_{odd}(\pi)$ the variation on the number of odd cycles in $G(\pi)$ and $G(\tau \cdot \pi)$, after the \emph{application} of a transposition $\tau$. Bafna and Pevzner~\cite{BafnaPevzner1998} demonstrated the following result.
	
	\begin{lemma}[Bafna and Pevzner~\cite{BafnaPevzner1998}]\label{lemma:bafna}
		$\Delta c_{odd}(\pi,\tau)  \in \{-2,0,2\}$.
	\end{lemma}
	
	A \emph{$\mu$-move} is a transposition $\tau$ such that $\Delta c_{odd}(\hat\pi, \tau)$ $=\mu$. Note that according to lemma above, the possible moves are $2$-move, $0$-move and $(-2)$-move.
	From Lemma~\ref{lemma:bafna}, Bafna and Pevzner~\cite{BafnaPevzner1998} derived the following lower bound
	
	\begin{theorem}[Bafna and Pevzner~\cite{BafnaPevzner1998}]
		\label{th:bafna}
		$d_t(\pi)\geq\frac{n+1-c_{odd}(\pi)}{2}$
	\end{theorem}
	
    The black edges of $G(\pi)$ can be numbered from $1$ to $n+1$ by assigning a label $i$ to each black edge $(-\pi_i, +\pi_{i-1})$. A $\kappa$-cycle $C$ visiting the black edges $i_1,\dots,i_{\kappa}$, in the order imposed by the cycle, can be written in $\kappa$ different ways, depending on the first black edge visited. If not otherwise specified, we will assume that the initial edge $i_1$ of $C$ is chosen as the greatest value, i.e., $i_1$ is such that $i_1 > i_s$, for all $s \in \{2,\dots,\kappa\}$. With this condition, if $i_1$, $\dots$, $i_\kappa$ is a decreasing sequence, $C$ is called an \emph{unoriented} cycle; otherwise $C$ is \emph{oriented}. Two pairs of black edges are said \emph{intersecting} if there are cycles $C=(\dots,a,b,\dots)$ and $D=(\dots,e,f,\dots)$ in $G(\pi)$ such that either $a > e > b > f$ or $e > a > f > b$. In this case, $C$ and $D$ are also said to be \emph{intersecting cycles}. Similarly, the triplets of black edges $(a,b,c)$ and $(d,e,f)$ are \emph{interleaving} if there are cycles $C=(\dots,a,b,c,\dots)$ and $D=(\dots,d,e,f,\dots)$ such that either $a > d > b > e > c > f$ or $d > a > e > b > f > c$. In such case, $C$ and $D$ are also said to be \emph{interleaving cycles}.
    
    \begin{figure*}[ht]
		\centering
		\includegraphics[scale=0.55]{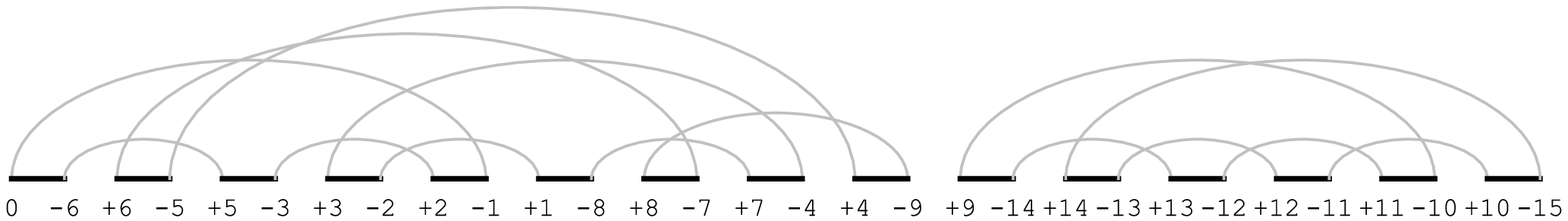}
        \caption{Cycle graph of $[6\;5\;3\;2\;1\;8\;7\;4\;9\;14\;13\;12\;11\;10]$.}
        \label{fig:cycles}
    \end{figure*}
    
    \begin{example}\label{ex4}
        The cycles $(5,3,1)$, $(8,6,4)$, $(15,13,11)$ and $(14,12,10)$ of $G([6\;5\;3\;2\;1\;8\;7\;4\;9\;14\;13\;12\;11\;10])$ (Figure~\ref{fig:cycles}) are unoriented, while $(9,2,7)$ is oriented. Furthermore, $(5,3,1)$ and $(8,6,4)$ are intersecting and the cycles $(15,13,11)$ and $(14,12,10)$ are interleaving.
    \end{example}
    
 	\subsubsection*{Simplification}
	
	Simplification is a technique introduced aiming to facilitate handling with long cycles of $G(\pi)$~\cite{HannPevzner1995}. It consists of inserting new elements, usually fractional numbers, into $\pi$ transforming it into a new \emph{simple permutation} $\hat\pi$, so that $G(\hat\pi)$ contains only short cycles. After the transformation, the elements of $\hat\pi$ can be mapped to consecutive integers. The positions of the new symbols can vary, but the insertion must be through \emph{safe transformations}.
	
	A transformation of $\pi$ into $\hat\pi$ is said to be \emph{safe} if, after the insertion of the new elements, the lower bound of Theorem~\ref{th:bafna} is maintained, i.e., $n(\pi) - c_{odd}(\pi) = n(\hat\pi) - c_{odd}(\hat\pi)$, where $n(\pi)$ and $n(\hat\pi)$ denote the number of black edges in $\pi$ and $\hat\pi$, respectively. If $\hat\pi$ is a permutation obtained from $\pi$ through safe transformations, then we say $\pi$ and $\hat{\pi}$ are \emph{equivalent}. Lin and Xue~\cite{lin2001signed} have shown that every permutation can be transformed into an equivalent simple one through safe transformations. A sorting of $\hat\pi$ can be mimicked to sort $\pi$ using the same number of transpositions~\cite{HannPevzner1995}.
	
	It is important to note that a permutation can be simplified in many different ways. Figure~\ref{fig:simple-permutation} shows the cycle graph of a possible simple permutation obtained by the simplification of $[4\;3\;2\;1\;8\;7\;6\;5]$ (Figure~\ref{fig:cycle-graph}). For a complete description of simplification and related results, the reader is referred to~\cite{HannPevzner1995,lin2001signed,hartman2006simpler}.
	
	\begin{figure*}[ht]
		\centering
		\includegraphics[scale=0.65]{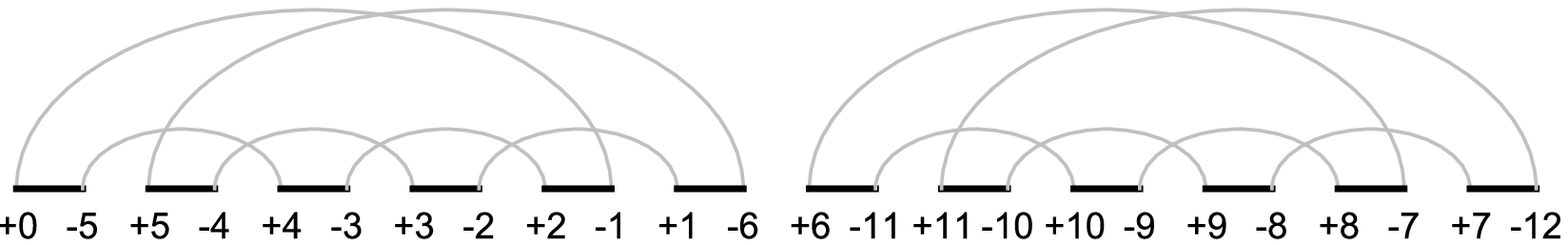}
		\caption{Cycle graph of $[5\;4\;3\;2\;1\;6\;11\;10\;9\;8\;7]$, mapped from $[4.1\;4\;3\;2\;1\;4.2\;8.1\;8\;7\;6\;5]$ using consecutive integers, obtained by the simplification of $[4\;3\;2\;1\;8\;7\;6\;5]$ (Figure~\ref{fig:cycle-graph}).}
		\label{fig:simple-permutation}
	\end{figure*}
	    
	\subsubsection*{Configurations and components}

    The concepts presented in this section were originally introduced by Elias and Hartman~\cite{EliasHartman2006} in the context of the simple permutations, with a special focus on the $3$-permutations, from which they derived their main results. As our work does not involve simplification, we modified some of them so that they could be extended to any permutation in $S_n$ and also to facilitate the correlation between the algebraic approach used in this work with the method of Elias and Hartman~\cite{EliasHartman2006}.

	A \emph{configuration} of cycles is a subgraph of $G(\pi)$ induced by one or more cycles. A configuration $A$ is \emph{connected}, if for any two cycles $C_1$ and $C_m$ of $A$, there are cycles $C_2,\dots,C_{m-1}$ in $A$ such that for each $i \in [1, m-1]$, $C_i$ intersects or interleaves with $C_{i+1}$. A \emph{component} is a configuration consisting of only one oriented cycle that does not intersect or interleave any other cycle of $G(\pi)$; or a maximal connected configuration in $G(\pi)$.
	
	Let $A$ be a configuration induced only by odd cycles. The \emph{3-norm} of $A$, denoted by $\norm{A}$, is the value $\frac{b-c(A)}{2}$, where $b$ is the number of black edges of $A$ and $c(A)$ is the number of cycles in $A$. If $\norm{A} \leq 8$, then $A$ is referred as being \emph{small}; otherwise, \emph{big}. The $3$-norm concept was not defined in Elias and Hartman~\cite{EliasHartman2006}. The intuition behind it is that it reflects the number of $3$-cycles a configuration containing cycles of arbitrary (odd) lengths would have if it were ``simplified''.
	
	\begin{example}
	    The $3$-norm of the configuration $\{(9,6$ $,8,2,4,1,3,5,7)\}$ from $G([4\;3\;2\;1\;8\;7\;6\;5])$ (Figure~\ref{fig:cycle-graph}) is $4$ and, consequently, it is a small configuration.
	\end{example}
	
	\begin{example}
	    The $3$-norms of the configurations $\{(7$ $,4,1), (8,5,2),(9,6,3)\}$ and $\{(14,12,10),(15,13,11)\}$ from $G([4\;8\;3\;7\;2\;6\;1\;5\;9\;14\;13\;12\;11\;10])$ (Figure~\ref{fig:eh-example-simpl}) are $3$ and $2$, respectively.
	\end{example}
	
	An \emph{open gate} is a pair of black edges $(a,b)$ of a cycle $C$ in $A$, such that one of its cyclic forms is $C=(a,b,\dots)$, that does not intersect with any other cycle in $A$ and there is no black edge $c$ in $C$, such that, if $a > b$, then $a$, $b$, $c$ is not a decreasing sequence; or, if $b > a$, then $b$, $c$, $a$ is not a decreasing sequence either. A configuration not containing open gates is called \emph{full configuration}.
	
	\begin{example}
	    The configurations $\{(7,4,1), (8,5,2),(9,$ $6,3)\}$ and $\{(14,12,10),(15,13,11)\}$ are small full components of $G([4\;8\;3\;7\;2\;6\;1\;5\;9\;14\;13\;12\;11\;10])$ (Figure~\ref{fig:eh-example-simpl}).
	\end{example}

	\subsubsection*{Sequences of transpositions}
	
    A sequence of transpositions $\tau_1$, $\dots$, $\tau_x$ is said to be a \emph{$(x,y)$-sequence}, for $x \geq y$, is a sequence of $x$ transpositions such that, at least $y$ of them are $2$-moves. A $(x,y)$-sequence is an \emph{$\frac{a}{b}$-sequence} if $\frac{x}{y} \leq \frac{a}{b}$ and $x \leq a$.
    
    \begin{example}
        The sequence $\tau_1=\rho(1,4,7)$, $\tau_2=\rho(2,5,8)$, $\tau_3=\rho(1,4,7)$, $\tau_4=\rho(3,6,9)$ is a $(4,3)$-sequence, which is also a $\frac{11}{8}$-sequence, for $[4\;8\;3\;7\;2\;6\;1$ $\;5\;9\;14\;13\;12\;11\;10]$ (Figure~\ref{fig:eh-example-simpl}).
    \end{example}
    
	\section*{The EH algorithm may require one extra transposition above the approximation of $1.375$}\label{sec:eh-exceed-1375}
	
	The first step of the EH algorithm is the simplification of the input permutation. In this section, we show that there are simplifications that, although producing equivalent simple permutations, causes the EH algorithm to require one extra transposition above the approximation of $1.375$. Two examples are explored next.
	
	Consider the permutation $\pi=[4\;3\;2\;1\;8\;7\;6\;5]$ shown in Figure~\ref{fig:cycle-graph}. The lower bound given by Theorem~\ref{th:bafna} is $4$, also its exact distance, corresponding to the application of four $2$-moves, shown in Figure~\ref{fig:3}.  One simplification of $\pi$ generates the permutation $[4.1\;4\;3\;2\;1\;4.2\;8.1\;8\;7\;6\;5]$, which mapped to consecutive integers is $\hat\pi=[5\;4\;3\;2\;1\;6\;11\;10\;9\;8\;7]$ (Figure~\ref{fig:simple-permutation}). Note that the lower bound of $\hat\pi$ is $4$ as well. However, there is no $\frac{11}{8}$-sequence to apply on $\hat\pi$. In fact, to optimally sort $\hat\pi$, two $(3,2)$-sequences are required. Therefore the EH algorithm using $\pi=[4\;3\;2\;1\;8\;7\;6\;5]$ as input, even applying an optimal sorting on $\hat\pi=[5\;4\;3\;2\;1\;6\;11\;10\;9\;8\;7]$, yields $6$ transpositions. However, the algorithm should require at most $5$ transpositions to not exceed the $1.375$-approximation ratio.
	
	\begin{figure}[ht]
		\centering
		$\rho(4,6,9)\cdot[4\;3\;2\;\underline{1\;8}\;\doubleunderline{7\;6\;5}]=[4\;3\;2\;7\;6\;5\;1\;8]$\\
		$\rho(3,5,8)\cdot[4\;3\;\underline{2\;7}\;\doubleunderline{6\;5\;1}\;8]=[4\;3\;6\;5\;1\;2\;7\;8]$\\
		$\rho(2,4,7)\cdot[4\;\underline{3\;6}\;\doubleunderline{5\;1\;2}\;7\;8]=[4\;5\;1\;2\;3\;6\;7\;8]$\\
		$\rho(1,3,6)\cdot[\underline{4\;5}\;\doubleunderline{1\;2\;3}\;6\;7\;8]=[1\;2\;3\;4\;5\;6\;7\;8]$
		\caption{Sorting $\pi=[4\;3\;2\;1\;8\;7\;6\;5]$ with $4$ transpositions.}~\label{fig:3}
	\end{figure}	
	
	The following example shows that, even if there are $\frac{11}{8}$-sequences of transpositions to apply on $\hat\pi$, the EH algorithm may require one transposition above the approximation ratio of $1.375$. Take the permutation $\pi'=[3\;6\;2\;5\;1\;4\;10\;9\;8\;7]$ (Figure~\ref{fig:eh-example}), with both the lower bound and distance equal to $5$, corresponding to the application of five $2$-moves, shown in Figure~\ref{fig:4}. A simplified version of $\pi'$ is $[3.1\;6.1\;3\;6\;2\;5\;1\;4\;6.2\;10.1\;10\;9\;8\;7]$, which mapped to consecutive integers is $\hat\pi'=[4\;8\;3\;7\;2\;6\;1\;5\;9\;14\;13\;12\;$ $11\;10]$ (Figure~\ref{fig:eh-example-simpl}). The EH algorithm sorts $\hat\pi'$ optimally by applying a $(4,3)$-sequence, followed by a $(3,2)$-sequence, in a total of $7$ transpositions. However, the algorithm should not require more than $6$ transpositions to not exceed the $1.375$-approximation ratio.
	
	\begin{figure*}[ht]
		\centering
		\includegraphics[scale=0.65]{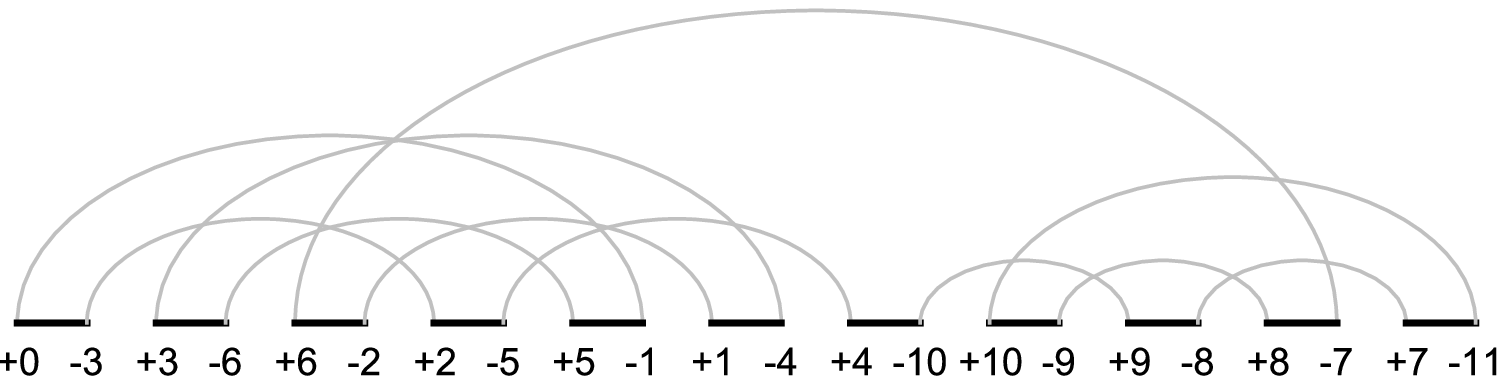}
		\caption{Cycle graph of $[3\;6\;2\;5\;1\;4\;10\;9\;8\;7]$.}
		\label{fig:eh-example}
	\end{figure*}
	
	\begin{figure*}[ht]
		\centering
		$\rho(6,8,11)\cdot[3\;6\;2\;5\;1\;\underline{4\;10}\;\doubleunderline{9\;8\;7}]=[3\;6\;2\;5\;1\;9\;8\;7\;4\;10]$\\
		$\rho(5,7,10)\cdot[3\;6\;2\;5\;\underline{1\;9}\;\doubleunderline{8\;7\;4}\;10]=[3\;6\;2\;5\;8\;7\;4\;1\;9\;10]$\\
		$\rho(3,6,9)\cdot[3\;6\;\underline{2\;5\;8}\;\doubleunderline{7\;4\;1}\;9\;10]=[3\;6\;7\;4\;1\;2\;5\;8\;9\;10]$\\
		$\rho(2,4,8)\cdot[3\;\underline{6\;7}\;\doubleunderline{4\;1\;2\;5}\;8\;9\;10]=[3\;4\;1\;2\;5\;6\;7\;8\;9\;10]$\\
		$\rho(1,3,5)\cdot[\underline{3\;4}\;\doubleunderline{1\;2}\;5\;6\;7\;8\;9\;10]=[1\;2\;3\;4\;5\;6\;7\;8\;9\;10]$
		\caption{Sorting $\pi'=[3\;6\;2\;5\;1\;4\;10\;9\;8\;7]$ with $5$ transpositions.}\label{fig:4}
	\end{figure*}
	
	\begin{figure*}[ht]
		\centering
		\includegraphics[scale=0.55]{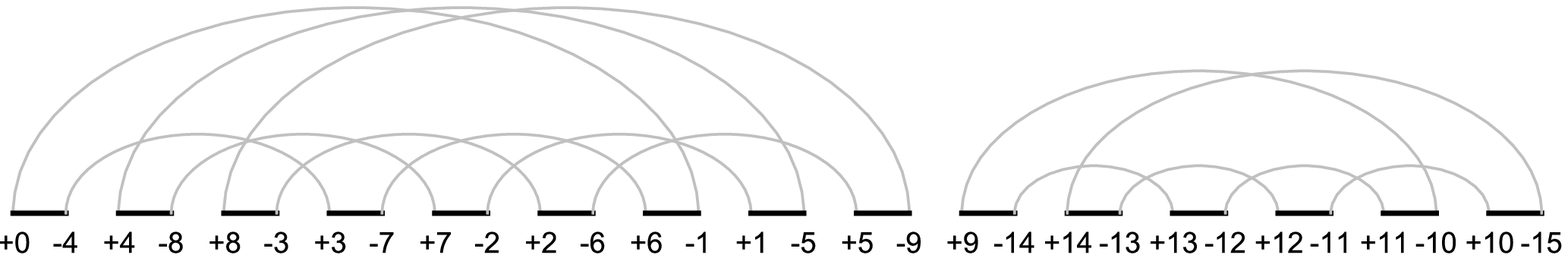}
		\caption{Cycle graph of $[4\;8\;3\;7\;2\;6\;1\;5\;9\;14\;13\;12\;11\;10]$, mapped from $\hat\pi'=[3.1\;6.1\;3\;6\;2\;5\;1\;4\;6.2\;10.1\;10\;9\;8\;7]$ using consecutive integers.}
		\label{fig:eh-example-simpl}
	\end{figure*}
	
	In both examples above, an initial $(2,2)$-sequence is ``missed'' during the simplification process. This sequence is essential to guarantee the $1.375$ approximation ratio when \emph{bad small components} remain in $G(\hat\pi)$ after the application of a number of $\frac{11}{8}$-sequences (Theorem 22~\cite{EliasHartman2006}). These are small full configurations which do not allow the application of $\frac{11}{8}$-sequences. It is important to stress that the extra transposition will be necessary regardless of the number of bad small components remaining in the cycle graph after applying a sequence of $\frac{11}{8}$-sequences (any number of), as long as the total number of remnant $3$-cycles is less than $8$ and the initial $(2,2)$-sequence that possibly existed initially, was ``missed'' during the simplification.
	
	It was already known by the literature that simplification maintained the lower bound, but not the transposition distance. However, it was not known that the simplification could have the effect of missing an initial $(2,2)$-sequence. In principle, the EH algorithm could be modified to guarantee the $1.375$-approximation ratio, and no extra transposition, by looking for the $(2,2)$-sequence in its first step, applying it case it exists, and only then simplifying the resulting permutation. However, using the already known techniques, this new ``modified'' EH algorithm would not keep the original time complexity of $O(n^2)$.

	\section*{Permutation groups}\label{sec:permutation-groups}
	
	The purpose of this section is to provide a background on permutation groups, necessary to understand the algebraic approach used in our method. Note that the results presented next are classical in the literature and their proofs are omitted since they can be found in abstract algebra textbooks~\cite{dummit2004abstract,Gallian2009}.
	
	The Symmetric Group $S_n$ on a finite set $E$ of $n$ symbols is the group formed by all \emph{permutations} on $n$ distinct elements of $E$, defined as bijections from $E$ to itself, under the operation of composition. The product of two permutations is defined as their composition as functions. Thus, if $\alpha$ and $\beta$ are permutations in $S_n$, then $\alpha \cdot \beta$, or simply $\alpha\beta$, is the function that maps any element $x$ of $E$ to $\alpha(\beta (x))$.
	
	An element $x \in E$ is said to be a \emph{fixed} element of $\pi \in S_n$, if $\pi(x) = x$. If there exists a subset $\{c_1, c_2, \dots, c_{\kappa-1}, c_\kappa\}$ of distinct elements of $E$, such that \[\pi(c_1)= c_2, \pi(c_2)= c_3, \dots, \pi(c_{\kappa-1})= c_\kappa, \pi(c_\kappa)= c_1,\] and $\pi$ fixes all other elements, then we call $\pi$ a \emph{cycle}. In \emph{cycle notation}, this cycle is written as $\pi=(c_1\;c_2\dots c_{\kappa-1}\;c_\kappa)$, but any of $(c_2\dots c_{\kappa-1}\;c_\kappa\;c_1)$, \dots, $(c_\kappa\;c_1\;c_2\dots c_{\kappa-1})$ denotes the same cycle $\pi$. The number \emph{$\kappa$} is the \emph{length} of $\pi$, also denoted as $|\pi|$. In this case, $\pi$ is also called a \emph{$\kappa$-cycle}.
	
	The \emph{support of a permutation $\alpha$}, denoted $Supp(\alpha)$, is the subset of moved (not fixed) elements of $E$. Two permutations $\alpha$ and $\beta$ are said \emph{disjoint}, if $Supp(\alpha) \cap Supp(\beta)=\varnothing$, i.e, if every symbol moved by one is fixed by the other. It is known that, if $\alpha$ and $\beta$ are disjoint, then they commute as elements of $S_n$, under the composition operation.

	\begin{lemma} 
		\label{lem:PermutationRepresentation}
		Every permutation in $S_n$ can be written as a product of disjoint cycles. This representation, called \emph{disjoint cycle decomposition}, is unique, regardless of the order in which the cycles are written in the representation.
	\end{lemma}
	
	For the sake of simplicity, a cycle $\beta$ \emph{in} or \emph{of} a permutation $\alpha$ is a cycle in the disjoint cycle decomposition of $\alpha$.
	
	The identity permutation $\iota$ is the permutation fixing all elements of $E$. Fixed elements sometimes are omitted in the cycle notation. However, when necessary they are written as $1$-cycles.
	
	\begin{example}
        The permutation $\pi=[4\;8\;3\;7\;2\;6\;1\;5]$ ($G(\pi)$ depicted in Figure~\ref{fig:graph}), in cycle notation, is represented by $\cycle{1,4,7}\cycle{3}\cycle{2,8,5}$. In this case, $3$ is the only fixed element and could be omitted in this notation. We can say that $\pi$ can be written, in unique form, as a product of two disjoint 3-cycles. This permutation could be written as product of other cycles, but these cycles would not be disjoint. Furthermore, $\pi$ could be written as $\cycle{1,7}\cycle{1,4}\cycle{2,5}\cycle{2,8}$, using four $2$-cycles\footnote{A $2$-cycle is commonly referred to as transposition in the algebra literature. In order to avoid misunderstanding with the terminology, in this paper, ``transposition'' always refers to swapping two adjacent blocks of symbols in a permutation.}, and also as $\cycle{1,7}\cycle{4,7}\cycle{1,7}\cycle{4,7}\cycle{2,5}\cycle{2,8}$, using six $2$-cycles. 
	\end{example}
	
	\begin{figure*}[ht]
		\centering
		\includegraphics[scale=0.65]{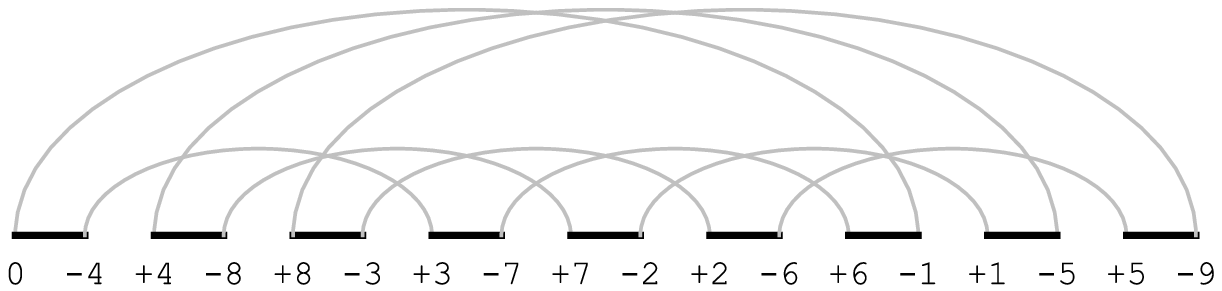}
		\caption{Cycle graph of $[4\;8\;3\;7\;2\;6\;1\;5]$}
		\label{fig:graph}
	\end{figure*}

	\begin{theorem}
		Every permutation in $S_n$ can be written as a (not unique) product of $2$-cycles.
	\end{theorem}
	
	A permutation $\alpha$ is said to be \emph{even}(\emph{odd}) if it can be written as a product of an even (odd) number of $2$-cycles. Next, we present some important results related to the parity of permutations.
	
	\begin{theorem}
		If a permutation $\alpha$ is written as a product of an even(odd) number of $2$-cycles, it cannot be written as a product of an odd(even) number $2$-cycles.
	\end{theorem}
	
	\begin{proposition}\label{prop:k-cycle-parity}
	    Let $\gamma$ be a $\kappa$-cycle. If $\kappa$ is odd, then $\gamma$ is an even permutation, otherwise $\gamma$ is odd.
	\end{proposition}
	
	\begin{theorem}\label{even-permutation}
		If $\alpha$, $\beta \in S_n$ are permutations with the same parity, then the product $\alpha\beta$ is even.
	\end{theorem}

	\section*{Algebraic formalisation of TDP}\label{sec:algebraic-formalization}
	
    A permutation $\pi=[\pi_1\;\pi_2\dots\pi_n]$ can be represented by many different ways. In the genome rearrangement context, where $\pi$ models a chromosome, an useful representation of $\pi$ is the set of cycles (defined using the labels of the black edges, as explained in the preliminaries) of $G(\pi)$. A naive way to obtain this representation is through building $G(\pi)$ itself. Another way to get the cycles of $G(\pi)$ is using the algebraic approach proposed by Mira el al.~\cite{Mira2008}, which is the one we employ in this paper. In this approach, we represent the permutation $\pi$ as the $(n+1)$-cycle $\epi=(0\;\pi_1\;\pi_2\dots\pi_n)$. Note that the ``dummy'' zero is used to allow the representation of the $n$ circular rotations of $\pi$ as distinct $(n+1)$-cycles. The next section will show a correspondence between the cycles of the permutation $\spi$ and the cycles of the cycle decomposition of $G(\pi)$, where $\eiota=\cycle{0,1,...,n}$\footnote{Note that $\eiota=\cycle{0,1,...,n}$ is not $\iota=(0)(1)\cdots(n)$.}.

	A $3$-cycle $\tau=\cycle{\pi_i,\pi_j,\pi_k}$ is said to be \emph{applicable} on $\epi$ if the symbols $\pi_i$, $\pi_j$ and $\pi_k$ appear in $\epi$ in the same cyclic order they are in $\tau$, i.e., $\epi=(\pi_i\dots\pi_j\dots\pi_k\dots)$~\cite{Mira2008}. The \emph{application} of $\tau$ on $\epi$ means multiply $\tau$ by $\epi$. Thus, and only in this case, the product $\tau \epi$ is a $(n+1)$-cycle, such that the symbols between $\pi_i$ and $\pi_{j-1}$, including $\pi_i$ but not $\pi_j$, in $\epi$ are ``cut'' and then ``pasted'' between $\pi_{k-1}$ and $\pi_k$, thus simulating a transposition on $\epi$, as $\tau\epi=\cycle{\pi_i,\pi_j,\pi_k}(\pi_0\pi_1\dots\pi_{i-1}\pi_i\pi_{i+1}\dots\pi_{j-1}\pi_j\pi_{j+1}\dots$ $\pi_{k-1}\pi_k\dots\pi_n)=(\pi_0\pi_1\dots\pi_{i-1}\pi_j\pi_{j+1}\dots\pi_{k-1}\pi_i\pi_{i+1}$ $\dots\pi_{j-1}\pi_k\dots\pi_n)$.

	\begin{example}
		Let $\epi=\cycle{0,4,3,2,1,8,7,6,5}$. The $3$-cycle $\tau=\cycle{0,2,7}$ is applicable to $\epi$ and thus simulates a transposition. The application $\tau\epi$ yields $\cycle{0,4,3,7,6,5,2,1,8}$. Now consider the $3$-cycle $\tau'=\cycle{0,1,2}$. Note that $\tau'$ is not applicable to $\epi$, and result of the product $\tau'\epi$ is $(0\;4\;3\;7\;6\;5)(1\;8)(2)$, which is not a $(n+1)$-cycle and therefore does not represent a chromosome in our approach.
	\end{example}
	
	Given two $(n+1)$-cycles $\epi$ and $\esigma$, the \textsc{Transposition Distance Problem} (TDP) consists of finding the minimum number $t$, denoted $d_t(\epi, \esigma)$, of transpositions represented as applicable $3$-cycles needed to transform $\epi$ into $\esigma$, i.e.,  
	\begin{align}
	\label{eq:1}
	\tau_t\dots\tau_1\epi & =\esigma.
	\end{align}
	
	From the equality above, multiplying both sides by $\epi^{-1}$, we have that 
	\begin{align}
	\label{eq:4}
	\tau_t\dots\tau_1 & =\esigma\epi^{-1}.
	\end{align}

    Observe that by Proposition~\ref{prop:k-cycle-parity} and Theorem~\ref{even-permutation}, the product of two cycles with the same length is an even permutation.
    
    \begin{proposition}\label{prop:spi}
        The permutation $\esigma\epi^{-1}$ is an even permutation.
    \end{proposition}
    
	The \emph{$3$-norm}~\cite{mira2005algebraic} of an even permutation $\alpha \in S_n$, denoted by $\norm{\alpha}_3$, corresponds to the smallest $\ell$ such that $\beta_{\ell}\dots\beta_1=\alpha$, where each $\beta_i$, $1 \leq i \leq \ell$, is a $3$-cycle. By Equation~\ref{eq:4}, as each transposition is a $3$-cycle, the $3$-norm of $\esigma\epi^{-1}$ is a lower bound for $d_t(\epi, \esigma)$.

	Denote by ${c^\circ}(\alpha)$ and ${c^\circ}_{odd}(\alpha)$, the number of cycles, including $1$-cycles; and the number of odd-length cycles (thus, even cycles), also including $1$-cycles, in $\alpha$, respectively. Mira and Meidanis~\cite{mira2005algebraic} demonstrated the following result.
	
	\begin{lemma}[Mira and Meidanis~\cite{mira2005algebraic}]\label{lem:norm}
		$\norm{\alpha}_3 = \frac{n-{c^\circ}_{odd}(\alpha)}{2}$.
	\end{lemma}
	
	Observe that by Proposition~\ref{prop:spi}, $\esigma\epi^{-1}$ is an even permutation. Therefore, as a corollary, a lower bound for TDP is derived.
	
	\begin{lemma}[Mira and Meidanis~\cite{mira2005algebraic}]\label{lem:lb-norm}
		If $\epi$ and $\esigma$ are $(n+1)$-cycles, then
		\begin{center}
			\begin{align*}
			d_t(\epi, \esigma) & \geq \norm{\esigma\epi^{-1}}_3 \\
			& \geq \frac{n+1-{c^\circ}_{odd}(\esigma\epi^{-1})}{2}.
			\end{align*}
		\end{center}
	\end{lemma}
	
    As already seen, TDP can be reduced to the problem of \textsc{Sorting By Transpositions} (SBT). In this case, $\esigma = \eiota = (0\;1\;2\dots n)$.
	
	In the next sections, we deal with the problem of SBT.
	
    \section*{New upper bound for SBT}
    
    We begin with some basic definitions and results concerning the $\spi$ permutation. Next, we present our main results which are a new upper bound for SBT and a $1.375$-approximation algorithm.

	\subsection*{The $\spi$ permutation}\label{sec:spi}

    An interesting fact is that the product $\spi$, in the algebraic approach, produces cycles corresponding to the same cycles of the cycle decomposition of the $G(\pi)$~\cite{Mira2008}. If we follow the edges of the cycles in $G(\pi)$, $\pi=[\pi_1 \; \pi_2 \; \dots \; \pi_n]$, taking note of the labels, without the sign, of the vertices where the gray edges enter and changing the label $-(n+1)$ to $0$, we get exactly the same cycles of $\spi$. 

	\begin{example}
	    Let $\pi=[4\;8\;3\;7\;2\;6\;1\;5\;9\;14\;13\;12\;11\;$ $10]$ ($G(\pi)$ represented in Figure~\ref{fig:cycles}). As seen in Example~\ref{ex4}, the cycles of $G(\pi)$ are $(5,3,1)$, $(8,6,4)$, $(9,2,7)$, $(14,12,10)$ and $(15,13,11)$. Now let $\epi=\cycle{0,4,8,3,7,2,6,1,5,9,14,13,12,11,10}$. The product $\spi=\cycle{0,11,13}\cycle{1,7,4}\cycle{2,8,5}\cycle{3,9,6}\cycle{10,12,14}$ has the same number of cycles as $G(\pi)$ and its cycles have the same relevant properties of the cycles of $G(\pi)$, such as length and orientation, and the relationships between them are also the same (interleaving, intersection, configurations, etc). These properties, in the algebraic approach, will be defined in this and in the next sections.
	\end{example}
	
	\subsubsection*{Cycles of $\spi$}

	Let $\gamma$ be a cycle in $\spi$. If $\gamma=(a\dots b\dots c\dots)$ and $\epi^{-1}=(a\dots c\dots b\dots)$, i.e., if the symbols $a$, $b$ and $c$ appear in $\gamma$ in a cyclic order that is distinct from the one in $\epi^{-1}$, then we say $(a,b,c)$ is an \emph{oriented triplet} and $\gamma$ is an \emph{oriented cycle}. Otherwise, if there is no oriented triplets in $\gamma$, then $\gamma$ is an \emph{unoriented cycle}. A cycle $\eta=(\eta_1\;\eta_2\dots\eta_{|\eta|})$ is a \emph{segment} of $\gamma$ if $\gamma=(\eta_1\;\eta_2\dots\eta_{|\eta|}\dots)$. Observe that by definition, a cycle in $\spi$ is a segment of itself. Analogously, we define a segment of a cycle $\gamma$ of $\spi$ as \emph{oriented} or \emph{unoriented}.
	
	Let $\delta=(a\;b\dots)$ and $\epsilon=(d\;e\dots)$ be two cycles of $\spi$. If $\epi^{-1}=(a\dots e\dots b\dots d\dots)$, i.e., if the symbols of the pairs $(a,b)$ and $(d,e)$ occur in alternate order in $\epi^{-1}$, we say these pairs \emph{intersect}, and that $\delta$ and $\epsilon$ are \emph{intersecting} cycles. A special case is when $\delta=(a\;b\;c\dots)$ and $\epsilon=(d\;e\;f\dots )$ are such that $\epi^{-1}=(a\dots e\dots b\dots f\dots c\dots d\dots)$, i.e., the symbols of the triplets $(a,b,c)$ and $(d,e,f)$ occur in alternate order in $\epi^{-1}$. In this case, $\delta$ and $\epsilon$ are said to be \emph{interleaving} cycles. Analogously, we define two segments of two $\spi$ cycles as \emph{intersecting} or \emph{interleaving}. 
	
	\begin{example}
		Let $\epi=\cycle{0,8,7,6,5,1,4,9,3,2}$ and $\spi=\cycle{0,3}\cycle{1,6,8}\cycle{2,4}\cycle{5,7,9}$. The cycles $\cycle{0,3}$ and $\cycle{2,4}$ are examples of intersecting cycles whereas $\cycle{1,6,8}$ and $\cycle{5,7,9}$ are interleaving cycles.
	\end{example}
		
	A $\kappa$-cycle in $\spi$ is called \emph{short} if $\kappa \leq 3$; otherwise, it is called \emph{long}. Similarly, a segment of a cycle of $\spi$ can be \emph{short} or \emph{long}.

	Observe that, from Equation~\ref{eq:4}, $\spi{\tau_1}^{-1}\dots{\tau_t}^{-1}=\iota$, i.e., the application of the transpositions $\tau_1$,$\dots$,$\tau_t$ sorting $\epi$ (i.e., transforming $\epi$ into $\eiota$) can be seen as the incremental multiplication of $\spi$ by ${\tau_1}^{-1}$, $\dots$, ${\tau_t}^{-1}$.
	
	Denote by $\Delta {c^\circ}(\spi, \tau)$ and by $\Delta {c^\circ}_{odd}(\spi, \tau)$, the differences ${c^\circ}(\spi \tau^{-1})-{c^\circ}(\spi)$ and ${c^\circ}_{odd}(\spi$ $\tau^{-1})-{c^\circ}_{odd}(\spi)$, respectively, where $\tau$ is an applicable $3$-cycle $\cycle{a,b,c}$. Depending on the symbols of $\tau$, its application on $\epi$ can affect the cycles of $\spi$ in the following distinct ways:
	\begin{enumerate}
		\item $a$, $b$ and $c$ are symbols belonging to the support of only one cycle $\gamma$ of $\spi$. We have two subcases:
		\begin{enumerate}
			\item If $a$, $b$ and $c$ appear in the same cyclic order in $\tau$ and $\gamma$, then $(a\dots b\dots c\dots )\cycle{a,b,c}^{-1}=(a\dots)(b\dots)(c\dots)$  i.e., the application of $\cycle{a,b,c}$ ``breaks'' $\gamma$ into $3$ shorter cycles. Thus, $\Delta {c^\circ}(\spi, \tau) = 2$.
			\item Otherwise, $\gamma=(a\dots c\dots b\dots)$. Then, \\ $(a\dots c\dots b\dots )\cycle{a,b,c}^{-1}$ $=$ $(a\dots b\dots c\dots)$ and $\Delta {c^\circ}(\spi, \tau) = 0$.
		\end{enumerate}
		\item $a$, $b$ and $c$ belong to the support of two different cycles $\gamma$ and $\delta$ of $\spi$. W.l.o.g., suppose $\gamma=(a\dots b\dots)$ and $\delta=(c\dots)$. Then we have that $(a\dots b\dots)(c\dots)\cycle{a,b,c}^{-1}=(a\dots c\dots)(b\dots)$ and $\Delta {c^\circ}(\spi, \tau) = 0$.
		\item $a$, $b$ and $c$ belong to the support of three different cycles $\gamma$, $\delta$ and $\epsilon$ of $\spi$. W.l.o.g., suppose $\gamma=(a\dots)$, $\delta=(b\dots)$ and $\epsilon=(c\dots)$. Then $(a\dots)(b\dots)(c\dots)\cycle{a,b,c}^{-1}=(a\dots c\dots b\dots)$ i.e., the application of $\cycle{a,b,c}$ ``joins'' the cycles $\gamma$, $\delta$ and $\epsilon$ into one longer cycle. Thus, $\Delta {c^\circ}(\spi, \tau) = -2$.
	\end{enumerate}
	
	From this observation and considering the possible parities of the cycles $\gamma$, $\delta$ and $\epsilon$, we have the following result.
	
	\begin{proposition}\label{prop}
		If $\tau$ is an applicable $3$-cycle then $\Delta {c^\circ}_{odd}(\spi, \tau)\in\{-2,0,2\}$.
	\end{proposition}
	
	The maximum number of cycles in $\spi$ is obtained if and only if $\spi$ is the identity permutation $\iota$. In this case, $\iota$ has $n+1$ cycles, being all even (odd-length) (in particular, they are all of length 1).
	
	We denote by \emph{$\mu$-move} an applicable $3$-cycle $\tau$ such that $\Delta {c^\circ}_{odd}(\spi, \tau)=\mu$. According to the Proposition~\ref{prop}, the possible moves are $(-2)$-move, $0$-move and $2$-move.
	
	\subsubsection*{Configurations and components}
	
	A \emph{configuration} $\Gamma$ is a product of segments of cycles of $\spi$, so that each cycle of $\spi$ has at most one segment in $\Gamma$. If $\norm{\Gamma}_3 \leq 8$ then $\Gamma$ is said to be \emph{small}; otherwise, \emph{big}.
	
	\begin{example} \label{ex:1}
		Let $\epi=(0\;6\;5\;3\;2\;1\;8\;7\;4\;9\;14\;13\;12\;11$ $\;10)$, so $\spi=\cycle{0,11,13}\cycle{1,3,6}\cycle{2,4,8}\cycle{5,7,9}\cycle{10,12,14}$. The product $\cycle{1,3,6}\cycle{2,4,8}$ is a small configuration of $\spi$.
	\end{example}
	
	A configuration $\Gamma$ is \emph{connected} if for any two segments $\gamma_1$ and $\gamma_m$ of $\Gamma$, there are segments $\gamma_2,\dots,\gamma_{m-1}$ in $\Gamma$ such that for each $i \in [1, m-1]$, $\gamma_i$ intersects or interleaves with $\gamma_{i+1}$. $\Gamma$ is said to be a \emph{component} if it consists of only one oriented cycle that does not intersect or interleave any other cycle of $\spi$; or a maximal connected configuration of $\spi$.

	\begin{example}
		Let $\epi=(0\;6\;5\;3\;2\;1\;8\;7\;4\;9\;14\;13\;12\;11$ $\;10)$. As $\spi=\cycle{0,11,13}\cycle{1,3,6}\cycle{2,4,8}\cycle{5,7,9}(10\;12\;$ $14)$, so $\cycle{0,11,13}\cycle{10,12,14}$ and $\cycle{1,3,6}\cycle{2,4,8}\cycle{5,7,9}$ are both components of $\spi$.
	\end{example}
	
    Let $\cycle{a,b,c}\cycle{d,e,f}$ be a configuration of $\spi$ consisting of two intersecting segments. If  $\epi^{-1}=(a\dots e\dots b\dots f\dots c\dots d\dots )$, i.e., if $\cycle{a,b,c}$ and $\cycle{d,e,f}$ interleave, then we call it the \emph{unoriented interleaving pair}. On the other hand, if $\epi^{-1}=(a\dots f\dots b\dots c\dots d\dots e\dots )$, i.e., $\cycle{a,b,c}$ and $\cycle{d,e,f}$ only intersect but do not interleave, then we call it the \emph{unoriented intersecting pair}.
	
	Let $\epsilon=(a\;b\dots)$ be a segment of a configuration $\Gamma$. We call the pair $(a,b)$ an \emph{open gate} in $\Gamma$, if there is no cycle $(c\;d\dots)$ in $\Gamma$ such that $(a,b)$ and $(c,d)$ intersect; and there is no $e \in Supp(\epsilon)$ such that $(a,b,e)$ is an oriented triplet. If $\Gamma$ is a configuration not containing open gates, then it is a \emph{full configuration}. Observe that the unoriented interleaving pair does not have open gates and therefore it is a full configuration. The unoriented intersecting pair, in its turn, has two open gates.
	
	\subsubsection*{Sequences of applicable $3$-cycles}
	
	We also denote by \emph{$(x,y)$-sequence}, for $x \geq y$, a sequence of $x$ applicable $3$-cycles $\tau_1$, $\dots$, $\tau_x$ such that, at least $y$ of them are $2$-moves. A $(x,y)$-sequence is said to be a \emph{$\frac{a}{b}$-sequence} if $x \leq a$ and $\frac{x}{y} \leq \frac{a}{b}$.
	
	\begin{example}
	    Let $\epi=(0\;4\;8\;3\;7\;2\;6\;1\;5\;9\;14\;13\;12\;11$ $\;10)$, so $\spi=\cycle{0,11,13}\cycle{1,7,4}\cycle{2,8,5}\cycle{3,9,6}(10\;12$ $\;14)$. The sequence $\tau_1=\cycle{1,4,7}$, $\tau_2=\cycle{2,8,5}$, $\tau_3=\cycle{1,4,7}$, $\tau_4=\cycle{3,9,6}$ is a $(4,3)$-sequence, which is also a $\frac{11}{8}$-sequence.
	\end{example}
	
	We say a configuration $\Gamma$ \emph{allows} the application of a $\frac{a}{b}$-sequence if it is possible to write this sequence using the symbols of $Supp(\Gamma)$.
	
	\subsection*{Auxiliary results}

    The proofs of some results in this section and the next rely on the analysis of a huge number of cases. Since it is impracticable to enumerate and verify by hand all the cases, we implemented, as Elias and Hartman~\cite{EliasHartman2006}, some computer programs~\cite{sourcecode} to systematically generate the proofs. In order to facilitate the visualisation and general understanding, the proofs are available to the reader in the form of a friendly web interface~\cite{proof}.

	Next we show some auxiliary results.
	
	\begin{proposition}\label{prop:3-cycle}
		If there is an oriented $3$-cycle $\gamma=\cycle{a,b,c}$ in $\spi$, then a $2$-move exists.
	\end{proposition}
	\begin{proof}
	    In this case, $\cycle{a,b,c}$ is a $2$-move.
	\end{proof}
	
	\begin{proposition}\label{prop:odd}
		If there is an odd (even-length) cycle in $\spi$, then a $2$-move exists.
	\end{proposition}
	\begin{proof}
		Since $\spi$ is an even permutation (Proposition~\ref{prop:spi}), then there is an even number of odd cycles in $\spi$. Let $\gamma=(a\;b\dots)$ and $\delta=(c\;d\dots)$ be two odd cycles of $\spi$. We have two cases:
		\begin{enumerate}
			\item $\gamma$ and $\delta$ intersect. In this case, we have that $\epi^{-1}=(a\dots d\dots b\dots c\dots)$. Then $(a\;b\;c)$ is a $2$-move.
			\item $\gamma$ and $\delta$ do not intersect. W.l.o.g, suppose $\epi^{-1}=(a\dots b\dots c\dots d\dots)$. In this case, $(a\;c\;b)$ is a $2$-move. Note that $\epi^{-1}=(a\dots c\dots d\dots b\dots)$ is not a distinct case, being equivalent to a cyclic rotation of $(a\dots b\dots c\dots d\dots)$ with the symbols $a$ and $b$ chosen differently.
		\end{enumerate}
	\end{proof}
	
	\begin{lemma}\label{lem:oriented-5-cycle}
		If there is a $5$-cycle $\gamma=\cycle{a,d,b,e,c}$ in $\spi$ such that $(a,b,c)$ is an oriented triplet, then there is a $2$-move or a $(3,2)$-sequence.
	\end{lemma}
	\begin{proof}
	    The possible distinct forms of $\epi$ relatively to the positions of the symbols of $Supp(\gamma)$ are listed below. For each one, there is either a $2$-move or a $(3,2)$-sequence.
	    \begin{enumerate}
            \item $\epi=(a\dots b\dots c\dots d\dots e\dots)$. $\tau_1=\cycle{a,b,c}$, $\tau_2=\cycle{b,c,d}$, $\tau_3=\cycle{c,d,e}$.
            \item $\epi=(a\dots b\dots c\dots e\dots d\dots)$. $\tau_1=\cycle{b,e,d}$.
            \item $\epi=(a\dots b\dots e\dots c\dots d\dots)$. $\tau_1=\cycle{a,e,c}$.
            \item $\epi=(a\dots e\dots b\dots d\dots c\dots)$. $\tau_1=\cycle{a,d,c}$.
            \item $\epi=(a\dots b\dots e\dots d\dots c\dots)$. $\tau_1=\cycle{a,d,c}$.
            \item $\epi=(a\dots d\dots b\dots e\dots c\dots)$. $\tau_1=\cycle{a,d,b}$.
	    \end{enumerate}
	\end{proof}
	
	Note that, by Lemma~\ref{lem:oriented-5-cycle}, if $\gamma=\cycle{a,d,b,e,c}$ is an oriented $5$-cycle in $\spi$ such that $(a,b,c)$ an oriented triplet, then $\epi=(a\dots b\dots c\dots d\dots e\dots)$ is the only form of $\epi$, relatively to the positions of the symbols of $Supp(\gamma)$, for which there is no $2$-move. In this case, we call $\gamma$ the \emph{bad oriented $5$-cycle}.
	
	\begin{lemma}\label{lem:even-oriented-k-geq-7}
		If there is an even (odd-length) $\kappa$-cycle $\gamma=(a\dots b\dots c\dots)$ in $\spi$ such that $\kappa\geq 7$ and $(a,b,c)$ is an oriented triplet, then there is either a $2$-move or $(4,3)$-sequence.
	\end{lemma}
	\begin{proof} If $(a\;b\;c)$ is a $2$-move, then the lemma holds. There is only one case where $\cycle{a,b,c}$ would not be a $2$-move. W.l.o.g, suppose that this case is 
		\[
		\gamma=(\underbrace{d\;e\dots b}_{odd}\;|\;\underbrace{f\dots c}_{even}\;|\;\underbrace{g\dots a}_{even}\;|).
		\]
		Vertical bars are used to indicate the locations where $\gamma$ would be broken if $\cycle{a,b,c}$ were applied on $\epi$, and subscripts to indicate the parity of the length of the resulting cycles. Note that the cycle $\gamma$ can be rewritten as the product 
		\[
		\gamma=(\underbrace{a\dots}_{odd})(\underbrace{b\dots}_{odd})(\underbrace{c\dots}_{odd})\cycle{a,d,e,b,f,c,g}.
		\]
		 There is only one form of $\epi$ relatively to the symbols of the support of $\cycle{a,d,e,b,f,c,g}$ not allowing the application of a $2$-move. It is $\epi=(a\dots e\dots f\dots g\dots d\dots b\dots$ $ c\dots)$. For this $\epi$, $\tau_1=(a\;e\;f)$, $\tau_2=(d\;e\;f)$, $\tau_3=(b\;f\;d)$, $\tau_4=(a\;c\;g)$ is $(4,3)$-sequence of transpositions.
	\end{proof}
	
	\begin{lemma}\label{lem:3-2}
		If $\spi\neq\iota$, then a $2$-move or $(3,2)$-sequence exists.
	\end{lemma}
	\begin{proof}
		If there is an odd (even-length) cycle in $\spi$, then by Proposition~\ref{prop:odd}, a $2$-move (i.e., a $(1,1)$-sequence) exists. Thus, we assume $\spi$ containing only even (odd-length) cycles. 
		
		\begin{enumerate}
			\item \label{lem:3-2:first} There is an oriented $\kappa$-cycle $\gamma$ in $\spi$. If $\kappa=3$, then Proposition~\ref{prop:3-cycle} gives a $2$-move and the lemma holds. If $\kappa=5$, then Lemma~\ref{lem:oriented-5-cycle} gives a $2$-move or $\gamma$ is the bad oriented $5$-cycle. In this case, there is a $(3,2)$-sequence. On the other hand, if $\kappa \geq 7$, then a $2$-move or a $(4,3)$-sequence, which contains a $(3,2)$-sequence, is given by Lemma~\ref{lem:even-oriented-k-geq-7}.
			
			\item All the cycles of $\spi$ are unoriented. Let $\gamma=\cycle{a,b,c}$ be a segment of a cycle of $\spi$. We have two cases:
			\begin{enumerate}
				\item $\gamma$ interleaves with another segment $\delta=\cycle{d,e,f}$. In this case, we have that $\epi=(a\dots f\dots c\dots e\dots b\dots d\dots )$. Then, $\tau_1=\cycle{a,c,b}$, $\tau_2=\cycle{d,e,f}$ and $\tau_3=\cycle{a,c,b}$ is a $(3,2)$-sequence.
				\item \label{lem:3-2-a} $\gamma$ intersects with two segments $\delta=\cycle{d,e,f}$ and $\epsilon=\cycle{g,h,i}$. For each of the $15$ distinct forms of $\epi$ (enumerated on~\cite{proof}), relatively to the possible positions of the symbols of $\gamma$, $\delta$ and $\epsilon$, there is a $(3,2)$-sequence.
			\end{enumerate}
		\end{enumerate}
	\end{proof}
	
	\subsection*{Configuration analysis}\label{sec:unoriented-confs}
	\label{sec:unoriented}

	At this point, we consider $\spi$ consisting only of even (odd-length) unoriented cycles of any size or bad oriented $5$-cycles. For the other cases, Propositions~\ref{prop:3-cycle},~\ref{prop:odd} and Lemma~\ref{lem:even-oriented-k-geq-7} give a $2$-move or a $(4,3)$-sequence.
	
	Our goal is to prove that, if $\norm{\spi}_3 \geq 8$, then a $\frac{11}{8}$-sequence of transpositions exists. The analysis is divided in two parts. In the first part, we analyse configurations obtained from basic ones (defined below) by extension. In the second part, we analyse $\spi$ composed only of small components, not allowing application of $\frac{11}{8}$-sequences. 
	
	\subsubsection*{Extension of basic configurations}
	
	The analysis starts with the bad oriented $5$-cycle, and the only two connected configurations of $3$-norm equal to $2$: the unoriented intersecting pair; and the unoriented interleaving pair. From these three \emph{basic configurations}, it is possible to build any other connected configuration of $\spi$ by successive extensions. From a configuration $\Gamma$, we can obtain a larger configuration $\Gamma'$, such that $\norm{\Gamma'}_3=\norm{\Gamma}_3+1$, extending $\Gamma$ by three different \emph{sufficient extensions}, as follows:
	
	\begin{enumerate}
		\item If $\Gamma$ has open gates, we can add a new unoriented $3$-cycle segment to $\Gamma$, closing at least one open gate.\label{extension:1}
		\item If $\Gamma$ has no open gates, we can add a new unoriented $3$-cycle segment to $\Gamma$, so that this segment intersects or interleaves another one in $\Gamma$.\label{extension:2}
		\item Let $\gamma$ be a segment in $\Gamma$. We can increase the length of $\gamma$ by $2$, originating a bad oriented $5$-cycle; or a longer unoriented segment, so that at least one open gate is closed, if $\Gamma$ has open gates; or creating up to two open gates, otherwise. \label{extension:3}
	\end{enumerate}
	
	\begin{example}
		We can extend the configuration $\Gamma$ of Example~\ref{ex:1} using extension~\ref{extension:1}, yielding $\Gamma'=\cycle{1,8,10}\cycle{5,7,12}\cycle{9,11,13}$. Then, with extension~\ref{extension:2}, we obtain $\Gamma''=\cycle{1,8,10}\cycle{2,4,6}\cycle{5,7,12}\cycle{9,11,13}$. Finally, with extension~\ref{extension:3}, we obtain $\Gamma'''=\cycle{0,3,5,7,12}\cycle{1,8,10}$ $\cycle{2,4,6}\cycle{9,11,13}$.
	\end{example}
	
	A \emph{sufficient configuration} is a configuration obtained by successively extending one of the basic configurations referred above. The computerised analysis proves the following result.
	
	\begin{lemma}\label{lem:sufficient}
		If it is possible to build a sufficient configuration $\Gamma$ of $\spi$ such that $\Gamma$ is big, then $\Gamma$ allows a $\frac{11}{8}$-sequence.
	\end{lemma}
	
	Observe that our definition of configuration extension is similar to the one devised by Elias and Hartman~\cite{EliasHartman2006}. However, Elias and Hartman~\cite{EliasHartman2006} only handled with configurations consisting of (unoriented) $3$-cycles, while our definition includes the generation of configurations containing long segments.
	
	Lemma~\ref{lem:sufficient} could be proven generating all the possible big configurations of $3$-norm equal to $9$ by extending the three basic configurations and then, for each, search for a $\frac{11}{8}$-sequence. However, this would be too time consuming. Instead, our computer program~\cite{sourcecode} employs a depth first search approach, in which, starting from the basic configurations, if we succeed in finding a $\frac{11}{8}$-sequence for a sufficient configuration, then we do not extend it further. The output of the program~\cite{sourcecode}, which proves Lemma~\ref{lem:sufficient}, is composed of 385,393 HTML files, one for each analysed case.

    \subsubsection*{Analysis of small full configurations which do not allow $\frac{11}{8}$-sequences}

    To conclude the analysis, now we handle the small full configurations for which the program~\cite{sourcecode} did not find $\frac{11}{8}$-sequences, and that can occur as small components in $\spi$. Small components not allowing $\frac{11}{8}$-sequences are called \emph{bad small components}.
    
    \begin{lemma}\label{lem:bad-small-components}
		The bad small components are the following:
		\begin{enumerate}
		    \item the bad oriented 5-cycle;
            \item the unoriented interleaving pair;
            \item the unoriented necklaces of size $4$, $5$ and $6$\footnotemark; and
            \item the twisted necklace of size $4$\footnotemark[\value{footnote}].
            \footnotetext{These components can be visualised on~\cite{proof}}
		\end{enumerate}
	\end{lemma}
	
	An \emph{unoriented necklace} of size $s$ is a component of $s$ unoriented $3$-cycles such that each cycle intersects with exactly two other cycles. The \emph{twisted necklace} of size $4$ is similar to the necklace of size $4$, but two of its cycles intersect with the three others.
	
	With the exception of the bad oriented 5-cycle, the bad small components listed above are the same ones found by Elias and Hartman~\cite{EliasHartman2006}, despite of the generation of configurations containing long segments in our analysis.
	
	With the help of computer program~\cite{sourcecode}, we prove the following result.
	
	\begin{lemma}\label{lem:bad-small-comp}
		If there is a configuration $\Lambda$ of $\spi$ consisting only of bad small components such that $\norm{\Lambda}_3 \geq 8$, then $\Lambda$ allows a $\frac{11}{8}$-sequence.
	\end{lemma}
	
	In order to prove Lemma~\ref{lem:bad-small-comp}, starting from each of the bad small components listed above, we successively extend them by adding another bad small component to the configuration, until finding a $\frac{11}{8}$-sequence. It turns out that no combination of bad small components with $3$-norm greater than $7$ was extended. The proof for Lemma~\ref{lem:bad-small-comp} is composed of 2,072 HTML files.
	
	\subsection*{New upper bound}

    The results presented in the previous section allow us to prove the corollary below. It follows from Proposition~\ref{prop:odd}, part~\ref{lem:3-2:first} from Lemma~\ref{lem:3-2}, which implies that, if we have an even (odd-length) oriented cycle in $\spi$, than we have a $2$-move, a $(4,3)$-sequence, or this cycle is the bad oriented $5$-cycle, and Lemmas~\ref{lem:sufficient} and~\ref{lem:bad-small-comp}.

	\begin{corollary}
	    If $\norm{\spi}_3 \geq 8$, then a $\frac{11}{8}$-sequence exists.
	\end{corollary}

	On the other hand, if $\norm{\spi}_3 < 8$, we only guarantee the existence of $\frac{3}{2}$-sequences. In the next section, we show that even in this scenario, the approximation ratio obtained by our algorithm is at most $1.375$.
	
	Finally, the last results prove the following upper bound for SBT.
	
	\begin{theorem}
	{\footnotesize
		\begin{align*}
		    d_t(\epi) \leq &\; 11 \Bigl\lfloor \frac{\norm{\spi}_3}{8} \Bigr\rfloor + \Bigl\lfloor\frac{3(\norm{\spi}_3 \bmod 8)}{2}\Bigr\rfloor \\
		    \leq &\; 11 \Bigl\lfloor \frac{n+1-{c^\circ}_{odd}(\spi)}{16} \Bigr\rfloor + \\
		    & \Bigl\lfloor\frac{3((n+1-{c^\circ}_{odd}(\spi)) \bmod 16)}{4}\Bigr\rfloor.
		\end{align*}}
	\end{theorem}
	
	Since ${c^\circ}_{odd}(\spi)=c_{odd}(\pi)$, the result above can be restated replacing $\epi$ and ${c^\circ}_{odd}(\spi)$, by $\pi$ and $c_{odd}(\pi)$ respectively. Thus, we derive the following upper bound for SBT, depending only on $n$ and $c_{odd}(\pi)$,
	
	\begin{theorem}
	{\footnotesize
    	\begin{align*}
		    d_t(\pi) \leq &\; 11 \Bigl\lfloor \frac{n+1-c_{odd}(\pi)}{16} \Bigr\rfloor + \Bigl\lfloor\frac{3((n+1-c_{odd}(\pi)) \bmod 16)}{4}\Bigr\rfloor.
		\end{align*}}
	\end{theorem}
	
	The new upper bound above improves the upper bound on the transposition distance devised by Bafna and Pevzner~\cite{BafnaPevzner1998}, valid for all $S_n$, based on their $1.5$-approximation algorithm~\cite{fertin2009combinatorics}. This upper bound allows us to obtain the following upper bound on the transposition diameter (TD). 
	\begin{corollary}\label{diameter}
		$TD(n) \leq 11 \Bigl\lfloor \frac{n}{16} \Bigr\rfloor + \Bigl\lfloor\frac{3(n \bmod 16)}{4}\Bigr\rfloor$
	\end{corollary}
	
	This upper bound on the transposition diameter, although tighter, for $n \geq 16$, than the one devised by Bafna and Pevzner~\cite{BafnaPevzner1998} of $\floor{\frac{3}{4} n}$ is not tighter than the one devised by Erikson et al.~\cite{eriksson2001sorting} of $\floor{\frac{2n - 2}{3}}$, for $n\ge 9$.
	
	\subsection*{The $1.375$-approximation algorithm}
	
	In this section, we present a $1.375$-approximation algorithm for SBT (Algorithm~\ref{alg1375}). For a permutation $\pi\in S_n$, the algorithm  returns an approximated distance between $\bar\pi$ and $\bar\iota$ or, equivalently, between $\pi$ and $\iota$. Intuitively, while $\norm{\spi}_3 \geq 8$, it repeatedly applies $\frac{11}{8}$-sequences of transpositions on $\epi$. When $\norm{\spi}_3 < 8$, the algorithm only guarantees the application of $\frac{3}{2}$-sequences.
	
	\begin{algorithm*}[ht!]
        \footnotesize
		\caption{A 1.375-approximation algorithm for SBT}
		\begin{algorithmic}[1]
			\Function{sbt1375}{$\epi$}
			\State $d \gets 0$
			\If {there is a $(2,2)$-sequence}  \label{(2,2)-seq}
    			\State apply a $(2,2)$-sequence
    			\State $d \gets d+2$
			\EndIf
			\While{there is an odd cycle in $\spi$}
    			\State apply a $2$-move  \Comment{Proposition~\ref{prop:odd}}
    			\State $d \gets d+1$
			\EndWhile
			\State {let $\Theta$ be the product of the unmarked cycles of $\spi$}
			\While{$\Theta \neq \iota$}\label{alg-main}
    			\If {there is a $2$-move from an oriented cycle of $\Theta$}
    				\State apply a $2$-move
    				\State $d \gets d+1$
    			\ElsIf{there is an even oriented $\kappa$-cycle in $\Theta$ such that $\kappa\geq 7$}
			        \State apply a $(4,3)$-sequence
    				\Comment{Lemma~\ref{lem:even-oriented-k-geq-7}}
    				\State $d \gets d+4$
    			\Else
					\State take a $3$-cycle segment $\gamma$ from a cycle of $\Theta$
					\State $\Gamma \gets \gamma$
					\State try to extend $\Gamma$ eight times
					\If {$\Gamma$ is big}
						\State apply a $\frac{11}{8}$-sequence of $x$ $3$-cycles \Comment{Lemma~\ref{lem:sufficient}}
						\State $d \gets d+x$
					\ElsIf {$\Gamma$ allows a $\frac{11}{8}$-sequence of $x$ $3$-cycles} 
						\State apply a $\frac{11}{8}$-sequence of $x$ $3$-cycles
						\State $d \gets d+x$
					\Else
						\State mark the cycles of $\Gamma$ \Comment{$\Gamma$ is a bad small component of $\spi$}
					\EndIf
    			\EndIf
    			\State let $\Lambda$ be the product of the marked cycles of $\spi$
				\If{$\norm{\Lambda}_3 \geq 8$}
				    \State unmark the cycles of $\Lambda$
                    \State apply a $\frac{11}{8}$-sequence of $x$ $3$-cycles \Comment{Lemma~\ref{lem:bad-small-comp}}
                    \State $d \gets d+x$
				\EndIf
			\EndWhile
			\While{$\spi \neq \iota$}
			    \State apply a $\frac{3}{2}$-sequence with $x$ $3$-cycles \Comment{Lemma~\ref{lem:3-2} (can be a $2$-move, i.e., a $(1,1)$-sequence, or a $(3,2)$-sequence)}
			    \State $d \gets d+x$
			\EndWhile
			\State \textbf{return} $d$
			\EndFunction
		\end{algorithmic}
		\label{alg1375}
	\end{algorithm*}

	To reach the intended approximation ratio of $1.375$ even when $\norm{\spi}_3 < 8$, the algorithm has to search for a $(2,2)$-sequence in its first step. In order to identify such a sequence, a look-ahead approach is used, meaning that the algorithm verifies if there is a second $2$-move, after applying a first $2$-move, generated either from an oriented cycle or from two odd (even-length) cycles of $\spi$.
	
	\begin{theorem}
		\label{th:complexity}
		The time complexity of Algorithm~\ref{alg1375} is $O(n^6)$.
	\end{theorem}
	\begin{proof}
		The time complexity of $O(n^6)$ is determined by the search for a $(2,2)$-sequence. In order not to miss a $2$-move, all triplets of an oriented cycle have to be checked to detect an oriented triplet leading to a $2$-move, which is $O(n^3)$. Finding a $2$-move by combining three symbols of two odd (even-length) cycles of $\spi$ requires $O(n^2)$. Thus, searching for a $(2,2)$-sequence with the look-ahead technique to check if there is an extra $2$-move needs time $O(n^6)$.
		
		The largest loop of the algorithm (line~\ref{alg-main}) needs time $O(n^4)$, while the last loop is $O(n)$.
	\end{proof}
	
	\begin{theorem}
		\label{th:ratio}
		Algorithm~\ref{alg1375} is a $1.375$-approximation algorithm for SBT.
	\end{theorem}
	\begin{proof}
		We note that this proof follows a very similar approach to the one used by Elias and Hartman~\cite{EliasHartman2006}. Let $f(x)=11 \Bigl\lfloor \frac{x}{8} \Bigr\rfloor + \Bigl\lfloor\frac{3(x \bmod 8)}{2}\Bigr\rfloor$. Depending on line~\ref{(2,2)-seq}, there are two cases.
		\begin{enumerate}
			\item There is a $(2,2)$-sequence. According to Lemma~\ref{lem:lb-norm}, it is not possible to sort $\epi$ using a sequence with less than $\norm{\spi}_3$ $2$-moves. Let $m = \norm{\spi}_3 - 2$ be the $3$-norm of $\spi$ after the application of a $(2,2)$-sequence. Algorithm~\ref{alg1375} sorts $\epi$ using a maximum of $f(m)+2$ transpositions, giving an approximation ratio of at most $\frac{f(m)+2}{m+2}$. In Table~\ref{tab1}, we can see that for all $0\leq r\leq 7$ such that $m=8l+r$ and $l \geq 0$, then $\frac{f(m)+2}{m+2}\leq \frac{11}{8}$.
			
			\item There is no $(2,2)$-sequence. If $\norm{\spi}_3 = 1$, then there is only one oriented $3$-cycle in $\spi$. In this case, there is a $2$-move and the theorem holds. Otherwise, we can raise the lower bound of Lemma~\ref{lem:lb-norm} by 1, since at least one $0$-move is required to sort $\epi$. Let $m=\norm{\spi}_3$. The approximation ratio given by Algorithm~\ref{alg1375} is at most $\frac{f(m)}{m+1}$. Table~\ref{tab1} also shows that for all $0\leq r\leq 7$ such that $m=8l+r$, $l \geq 0$, then $\frac{f(m)}{m+1}\leq \frac{11}{8}$.
		\end{enumerate}
	\end{proof}
	
	\begin{table*}[ht]
		\centering
		\caption{For all $0\leq r\leq 7$ such that $m=8l+r$ and $l \geq 0$, the approximation ratio given by Algorithm~\ref{alg1375} is at most $\frac{11}{8}=1.375$.}
		\label{my-label}
		\begin{tabular}{c|cccccccc}
			$r$                  & $0$                  & $1$                  & $2$                  & $3$                  & $4$                  & $5$                  & $6$                   & $7$                   \\ \hline
			$\frac{f(m)+2}{m+2}$ & $\frac{11l+2}{8l+2}$ & $\frac{11l+4}{8l+3}$ & $\frac{11l+5}{8l+4}$ & $\frac{11l+6}{8l+5}$ & $\frac{11l+8}{8l+6}$ & $\frac{11l+9}{8l+7}$ & $\frac{11l+11}{8l+8}$ & $\frac{11l+12}{8l+9}$ \\ \hline
			$\frac{f(m)}{m+1}$   & $\frac{11l}{8l+1}$   & $\frac{11l+2}{8l+2}$ & $\frac{11l+3}{8l+3}$ & $\frac{11l+4}{8l+4}$ & $\frac{11l+6}{8l+5}$ & $\frac{11l+7}{8l+6}$ & $\frac{11l+9}{8l+7}$  & $\frac{11l+10}{8l+8}$
		\end{tabular}
		\label{tab1}
	\end{table*}
	
	\section*{Results and discussion}
	
	We implemented Algorithm~\ref{alg1375} and the EH algorithm, having tested both using the Rearrangement Distance Database provided by GRAAu~\cite{galvao2015audit}. We computed all transposition distances using both algorithms for all permutations of size $n$, $2\leq n\leq 12$.
	
	It should be noted that Elias and Hartman~\cite{EliasHartman2006} did not provide a publicly available implementation of their algorithm, which we could use as a reference. To the best of our knowledge, the only implementation of the EH algorithm reported in the literature, without the use of heuristics, is the one of Dias and Dias~\cite{dias2010extending}, but their implementation is not publicly available either. This led us to implement the EH algorithm from scratch.
	
	It is noteworthy, that our implementation of the EH algorithm in closer to the version\footnote{As this version uses one single loop to apply $\frac{11}{8}$-sequences.} previously presented on WABI in 2005~\cite{elias20051}, since we found an issue in the algorithm outlined in~\cite{EliasHartman2006} (the algorithms are presented differently in both versions of their work). The issue has to do with the application of $\frac{11}{8}$-sequences when $G(\hat\pi)$ contains only bad small components. As presented on~\cite{EliasHartman2006}, once all bad small components are identified, the algorithm enters a loop and continuously applies $\frac{11}{8}$-sequences (given by their Lemma 17~\cite{EliasHartman2006}), until the number of cycles in $G(\hat\pi)$ is less than $8$. However, we found scenarios where the application of $\frac{11}{8}$-sequences given by Lemma 17~\cite{EliasHartman2006} can create small components in $G(\hat\pi)$ that are not bad, which can eventually prevent the application of the lemma in the next iterations. One such scenario is when we have a permutation consisting only of unoriented necklaces of size $6$~\cite{EliasHartman2006}, $3$ or more, side by side. To give an illustration, take $\hat\pi=[17\;16\;3\;2\;1\;6\;5\;4\;9\;8\;7\;12\;11\;10\;15\;14\;13\;18\;35\;34\;21\;20$ $\;19\;24\;23\;22\;27\;26\;25\;30\;29\;28\;33\;32\;31]$ whose $G(\hat\pi)$ consists precisely of two unoriented necklaces of size $6$ side by side. Since the sum of $3$-cycles is $12$, Lemma 17~\cite{EliasHartman2006} guarantees us the existence of a $\frac{11}{8}$-sequence. The $(11,8)$-sequence given by Elias and Hartman~\cite{EliasHartman2006} for this permutation (by combining two unoriented necklaces of size $6$ side by side) is $\tau_1=\rho(1,3,5)$, $\tau_2=\rho(7,11,26)$, $\tau_3=\rho(9,13,35)$, $\tau_4=\rho(4,10,34)$, $\tau_5=\rho(2,13,30)$, $\tau_6=\rho(1,18,20)$, $\tau_7=\rho(6,17,32)$, $\tau_8=\rho(5,14,22)$, $\tau_9=\rho(15,27,35)$, $\tau_{10}=\rho(18,28,36)$, $\tau_{11}=\rho(6,19,35)$. After applying this sequence, we have $\hat\pi=[1\;2\;3\;4\;20\;21\;22\;27\;26\;25\;30\;31$ $\;32\;11\;12\;13\;14\;15\;16\;17\;18\;19\;24\;23\;5\;6\;7\;8\;9\;10\;29\;28$ $\;33\;34\;35]$. Observe that now $G(\hat\pi)$ contains a small component of four unoriented $3$-cycles that despite being small, is not bad.
	
	To overcome the issue above, our solution was to apply a $\frac{11}{8}$-sequence given by Lemma 17~\cite{EliasHartman2006} as soon as the sum of the number of $3$-cycles of the the bad small components, as they are identified in the main loop, is greater than $7$, inside the loop itself (line $5$ of the algorithm outlined in~\cite{EliasHartman2006}), as opposed to its position within a loop of its own (line 6~\cite{EliasHartman2006}). Similar solution was employed by our Algorithm~\ref{alg1375}.
	
	We found another issue in the last loop of both versions of the EH algorithm (\cite{EliasHartman2006},\cite{elias20051}). It is not always possible to apply a $(3,2)$-sequence at that point. Sometimes, only a $2$-move exists, as the Lemma $7$~\cite{EliasHartman2006} itself states. Take, for instance, the permutation $\hat\pi=[14\;13\;3\;2\;1\;6\;5\;4\;9\;8\;7\;12\;11\;10]$ whose $G(\hat\pi)$ consists of an unoriented necklace of size $5$. Observe that there is no $\frac{11}{8}$-sequence to apply on $\hat\pi$. In the last loop (\cite{EliasHartman2006},\cite{elias20051}), Elias and Hartman~\cite{EliasHartman2006} gives two $(3,2)$-sequences $\tau_1=\rho(1,10,14)$, $\tau_2=\rho(4,6,15)$, $\tau_3=\rho(3,5,14)$, $\tau_4=\rho(4,8,9)$, $\tau_5=\rho(2,5,8)$, $\tau_6=\rho(1,3,6)$. After applying this sequence, we have $\hat\pi=[1\;6\;7\;8\;2\;3\;4\;5\;9\;10\;11\;12\;13\;14]$ whose $G(\hat\pi)$ contains only one oriented $3$-cycle, making it impossible to apply a further $(3,2)$-sequence. In this particular case, the $2$-move $\rho(2,5,9)$ concludes the sorting of $\hat\pi$. Our implementation~\cite{sourcecode} of the EH algorithm includes ``fixes'' for both issues described above.
	
	Continuing with the analysis of the results for the permutations of maximum length $12$, as presented by Table~\ref{tab:1}, the approximation ratio obtained by the EH algorithm exceeds $1.375$. On the other hand, our proposed algorithm does not exceed the ratio of $1.333\bar 3$. However, we presume that approximations of $1.375$ would appear for permutations in $S_n$, $n \geq 16$, since in order to exist an $(11,8)$-sequence, $Supp(\spi)$ has to have at least $17$ symbols.

	\begin{table*}[ht]
		\begin{center}
			\caption{Comparison of the maximum approximation ratios given by the EH algorithm with ours (Alg1). The table includes other metrics such as the average approximation ratio and average distance given by each algorithm and the number of times the EH algorithm exceeds the $1.375$-approximation ratio as well as the time consumed by each algorithm to sort all permutations of each size. Decimal values are truncated to 4 places.} 
			\begin{threeparttable}
			\begin{tabular}{c|c|c|c|c|c|c|c|c|c|c}
				\hline
				& \multicolumn{1}{|c|}{\textbf{Transposition}} & \multicolumn{2}{|c|}{\textbf{Max. approx.}} & \multicolumn{2}{|c|}{\textbf{Average}} & \multicolumn{2}{|c|}{\textbf{Average}}  & \multicolumn{1}{|c|}{\textbf{Number of times}}           & \multicolumn{2}{|c}{\textbf{Time to sort all}} \\
				& \multicolumn{1}{|c|}{\textbf{diameter}} & \multicolumn{2}{|c|}{\textbf{ratio}}        & \multicolumn{2}{|c|}{\textbf{approx. ratio}} & \multicolumn{2}{|c|}{\textbf{distance}} & \multicolumn{1}{|c|}{\textbf{EH exceeded the}} & \multicolumn{2}{|c}{\textbf{permutations}\tnote{*}} \\
				\cline{3-4}\cline{5-6}\cline{7-8}\cline{10-11}
				\textbf{n} & & \textbf{EH} & \textbf{Alg1} & \textbf{EH} & \textbf{Alg1} & \textbf{EH} & \textbf{Alg1} & \textbf{$1.375$-approx.} & \textbf{EH} & \textbf{Alg1} \\
				\hline    
				$2$  & $1$ & $1.00$        & $1.00$        & $1.0$    & $1.0$    & $1.00$   & $1.00$   & $0$    & $<1s$           & $<1s$ \\ 
				$3$  & $2$ & $1.00$        & $1.00$        & $1.0$    & $1.0$    & $1.20$   & $1.20$   & $0$    & $<1s$           & $<1s$\\ 
				$4$  & $3$ & $1.00$        & $1.00$        & $1.0$    & $1.0$    & $1.6086$ & $1.6086$ & $0$    & $<1s$           & $<1s$\\ 
				$5$  & $3$ & $1.00$        & $1.00$        & $1.0$    & $1.0$    & $2.0924$ & $2.0924$ & $0$    & $<1s$           & $<1s$\\ 
				$6$  & $4$ & $1.333\bar 3$ & $1.00$        & $1.0004$ & $1.0$    & $2.6063$ & $2.6050$ & $0$    & $<1s$           & $<1s$\\ 
				$7$  & $5$ & $1.333\bar 3$ & $1.25$        & $1.0129$ & $1.0113$ & $3.1762$ & $3.1704$ & $0$    & $<1s$           & $<1s$\\ 
				$8$  & $6$ & $1.5$         & $1.25$        & $1.0210$ & $1.0183$ & $3.7178$ & $3.7076$ & $2$    & $<2s$           & $<2s$\\ 
				$9$  & $6$ & $1.5$         & $1.25$        & $1.0301$ & $1.0256$ & $4.2796$ & $4.2603$ & $20$   & $\approx 10s$   & $\approx 13s$ \\ 
				$10$ & $7$ & $1.5$         & $1.25$        & $1.0341$ & $1.0282$ & $4.8051$ & $4.7772$ & $110$  & $\approx 3m$    & $\approx 2m$ \\ 
				$11$ & $8$ & $1.5$         & $1.333\bar 3$ & $1.0392$ & $1.0321$ & $5.3526$ & $5.3157$ & $440$  & $\approx 35m$   & $\approx 30m$ \\ 
				$12$ & $9$ & $1.5$         & $1.333\bar 3$ & $1.0415$ & $1.0336$ & $5.8694$ & $5.8248$ & $1448$ & $\approx 8.5h$ & $\approx 8.1h$ \\ 
				\hline
			\end{tabular}
			\begin{tablenotes}\footnotesize
                \item [*] The permutations of each size were sorted in parallel using a pool of $8$ threads.
            \end{tablenotes}\label{tab:1}
			\end{threeparttable}
		\end{center} 
	\end{table*}
	
	We also compared (Table~\ref{tab:3}) the percentage of computed distances that are equal to transposition distance outputted by our algorithm and EH's with others available in the literature. In particular, we added to the comparison an algorithm with a higher $1.5$-approximation, but with good results~\cite{walter2005improving}; one using similar algebraic approach~\cite{Mira2008}, also $1.5$-approximation; and another one that also uses an EH-like strategy with approximation ratio of $1.375$~\cite{dias2010improved}. 
	
	\begin{table*}[ht]
		\begin{center}
			\caption{Comparison of the percentage of computed distances that are equal to transposition distance, given by different algorithms (WDM~\cite{Walter:2000:NAA:829519.830850}, M~\cite{Mira2008}, BPwh~\cite{walter2005improving} and DD~\cite{dias2010improved}), in comparison to the EH algorithm and ours. Decimal values are truncated to $2$ places.} 
			\begin{tabular}{c|c|c|c|c|c|c}
				\hline    
				\textbf{n} & \textbf{WDM} & \textbf{M} & \textbf{BPwh} & \textbf{DD} & \textbf{EH}  & \textbf{Alg1} \\
				\hline    
                    $2$   & -       & $100.00$ & $100.00$ & -        & $100.00$ & $100.00$ \\
                    $3$   & -       & $100.00$ & $100.00$ & -        & $100.00$ & $100.00$ \\
                    $4$   & -       & $100.00$ & $100.00$ & $100.00$ & $100.00$ & $100.00$ \\
                    $5$   & -       & $100.00$ & $100.00$ & $100.00$ & $100.00$ & $100.00$ \\
                    $6$   & $99.17$ & $100.00$ & $100.00$ & $100.00$ & $99.86$  & $100.00$ \\
                    $7$   & $98.58$ & $100.00$ & $100.00$ & $100.00$ & $94.90$  & $95.47$ \\
                    $8$   & $97.11$ & $99.69$  & $99.91$  & $100.00$ & $91.64$  & $92.65$ \\
                    $9$   & $96.05$ & $99.17$  & $99.72$  & $99.99$  & $86.62$  & $88.54$ \\
                    $10$  & $94.12$ & $98.09$  & -        & $99.97$  & $83.80$  & $86.53$ \\
                    $11$  & $92.81$ & $96.90$  & -        & -        & $79.40$  & $82.98$ \\
                    $12$  & -        & -       & -        & -        & $76.67$  & $80.91$ \\
                    \hline
			\end{tabular}\label{tab:3}
		\end{center} 
	\end{table*}
	
	As shown by Table~\ref{tab:3}, regarding the percentage of computed distances that are equal to transposition distance metric, the best algorithm seems to be the algorithm of Dias and Dias~\cite{dias2010improved}, although they do not present results for $n > 10$. Importantly, this algorithm employs several heuristics, some introduced by a previous work~\cite{dias2010extending}, to improve the performance of the EH algorithm. One of these heuristics is exactly a search for a second $2$-move using a look-ahead technique. However, it is not clear whether the heuristic employed by them never miss a $(2,2)$-sequence if it exists. Furthermore, Dias and Dias~\cite{dias2010improved} does not state the complexity of their algorithm, but we believe that, by analyzing the algorithm~\cite{dias2010extending} which they were based on, the time complexity is higher than $O(n^3)$.
	
	The performance of our algorithm and EH's were also investigated for longer permutations. For this, we created a dataset of longer permutations with sizes ranging from $20$ to $500$ (incremented by $10$). For each of the $49$ sets, $1,000$ instances were randomly generated and sorted using both algorithms. Figure~\ref{fig:ratios} shows the maximum and the average approximation ratios obtained from both ones. It should be noted that the approximation ratios were calculated in relation to the lower bound given by Theorem~\ref{th:bafna}, since is impracticable to calculate the exact distance for such long permutations. Similar experiment was conducted by Dias and Dias~\cite{dias2010extending}, but in their experiment, they worked with smaller sets, also ranging from $20$ to $500$ (incremented by $10$), but containing only $100$ instances. By comparing the results, we may conclude that our algorithm and theirs achieve similar results. Dias and Dias~\cite{dias2010improved} also conducted experiments with longer permutations, but with sizes ranging only from $10$ to $100$ (incremented by $10$), where each set contained $100$ instances, and collected the running times. We may conclude, by comparing the results presented in their paper, that our algorithm performs better than theirs.

	\begin{figure}[ht]
	    \begin{tikzpicture}
            \begin{axis}[
                scaled y ticks = false,
                y tick label style={/pgf/number format/fixed,
                /pgf/number format/1000 sep = \thinspace},
                xmin=20,xmax=500,
                ymin=1, ymax=1.34,
                ymajorgrids=true, 
                ytick distance=0.05,
                xtick distance=20,
                tick label style={font=\scriptsize},
                axis y line*=left, axis x line*=bottom,
                x tick label style={font=\scriptsize,rotate=90,anchor=east}]
                \addplot[color=red,legend entry=\scriptsize max.\; Alg1,smooth,thick] table[col sep=comma,header=false,
                x index=0,y index=1] {ratio.csv};
                \addplot[color=blue,legend entry=\scriptsize average\;Alg1,thick] table[col sep=comma,header=false,
                x index=0,y index=2] {ratio.csv};
                \addplot[color=orange,legend entry=\scriptsize max.\;EH,smooth,thick] table[col sep=comma,header=false,
                x index=0,y index=3] {ratio.csv};
                \addplot[color=green,legend entry=\scriptsize average\;EH,smooth,thick] table[col sep=comma,header=false,
                x index=0,y index=4] {ratio.csv};
            \end{axis}
        \end{tikzpicture}
	    \caption{Average and maximum approximation ratios obtained for each size in our dataset of longer permutations.} 
	    \label{fig:ratios}
	\end{figure}

	Figure~\ref{fig:times} shows how much time each algorithm (ours and EH's) took to sort all the $1,000$ instances of each of the $49$ sets. The results presented by this figure show that, despite having a high time complexity, our algorithm has good performance in practice, even outperforming EH's.

	\begin{figure}[ht]
        \begin{tikzpicture}
            \begin{axis}[
                scaled y ticks = false,
                y tick label style={/pgf/number format/fixed,
                /pgf/number format/1000 sep = \thinspace},
                xmin=20,xmax=500,
                ymin=0, ymax=26.5,
                ymajorgrids=true, 
                ytick distance=2,
                xtick distance=20,
                tick label style={font=\scriptsize},
                axis y line*=left, axis x line*=bottom,
                legend pos=north west,
                x tick label style={font=\scriptsize,rotate=90,anchor=east}]
                \addplot[color=red,legend entry=\scriptsize Alg1,smooth,thick] table[col sep=comma,header=false,
                x index=0,y index=1] {times.csv};
                \addplot[color=blue,legend entry=\scriptsize EH,thick] table[col sep=comma,header=false,
                x index=0,y index=2] {times.csv};
            \end{axis}
        \end{tikzpicture}
        \caption{Time in minutes each algorithm took to sort all the instances of each size of our dataset of longer permutations.} \label{fig:times}
    \end{figure}
    
	The dataset of longer permutations used in our experiments, the statistics computed, as well as the source code of the implementation of the EH algorithm and ours are available at~\cite{sourcecode}. All experiments have run on a computer equiped with a Core i7 vPro $8^{th}$ Gen processor, with 4 cores and 8 threads, and 48GB of RAM.

	\section*{Conclusions}
	
	In this paper, we first showed that the EH algorithm may require one extra transposition above the $1.375$-approximation ratio, depending on how the input permutation is simplified. This occurs when there is a first $(2,2)$-sequence in the original permutation that is ``missed'' during simplification, and bad small components remain in the cycle graph after the application of any number of $\frac{11}{8}$-sequences.
	
	Then, we proposed a new upper bound for the transposition distance which holds for all $S_n$. Next, we proposed a $1.375$-approximation algorithm to solve SBT based on an algebraic approach that does not employ simplification and which guarantees the $1.375$-approximation ratio for all $S_n$. To the best of our knowledge, this is the first algorithm guaranteeing an approximation ratio below $1.5$ not using simplification. 
	
	We also pointed out an issue regarding the application of $\frac{11}{8}$-sequences when the cycle graph contains only bad small components in the EH algorithm outlined in~\cite{EliasHartman2006} and another one, related to the application of $(3,2)$-sequences when there is no $\frac{11}{8}$-sequences to apply, that affects both versions of the algorithm outlined in~\cite{EliasHartman2006} and~\cite{elias20051}.
	
	Implementations of the EH algorithm and ours were tested against permutations of maximum length of $12$. The results showed that our algorithm does not exceed the $1.375$-approximation ratio and produces a higher percentage of computed distances that are equal to transposition distance, when compared to those computed by the EH algorithm. These percentages were also compared to others available in the literature. Considering this metric, the algorithm with the best results seems to be the one of Dias and Dias~\cite{dias2010improved}, although they do not present results for $n > 10$.
	
	We conducted an experiment involving longer permutations of maximum length $500$. The results showed that our algorithm outperforms the EH algorithm, both in relation to the approximation ratios obtained and running times. Still on the longer permutations, our algorithm seems to be comparable to the one of Dias and Dias~\cite{dias2010extending}, when we consider the approximation ratios obtained by both. Regarding the running times, Dias and Dias~\cite{dias2010extending} also performed some simulations for permutations with a maximum size of $100$. Considering only the results for permutations with this maximum size, our algorithm seems faster.
	
	It is noteworthy that the time complexity of our algorithm is high. A possible future work could the investigation of a more efficient way to find a $(2,2)$-sequence at the beginning of our algorithm. Following a different direction, another future work could be the investigation of ``good'' simplifications, i.e., simplifications that do not have the effect of missing a $(2,2)$-sequence when it exists. We have no clue whether such a ``good'' simplification always exists or not. In any case, we have the intuition that to find it, if it exists, the computational cost would be the same as searching for a $(2,2)$-sequence.
	
	The experiment with small permutations of maximum length $12$ showed that the percentages of computed distances by our algorithm that are equal to transposition distance are low compared to others in the literature. A possible way to improve the results would be investigating the adoption of heuristics.
	
	Finally, we intend to use the algebraic approach presented in this paper to study and solve other rearrangement events affecting one chromosome, e.g., reversals and block-interchange.
	
    \begin{backmatter}
    
    \section*{Competing interests}
    The authors declare that they have no competing interests.

    \section*{Author's contributions}
    First draft: LAGS, MEMTW, NRR. Proofs and algorithm implementation: LAGS. Final manuscript: LAGS, LABK, MEMTW. All authors read and approved the final manuscript.

	\section*{Acknowledgements}
	
	The authors kindly thank Isaac Elias for the invaluable discussion. They also thank Annachiara Korchmaros, and the anonymous reviewers, whose comments helped to improve this manuscript. MEMTW thanks CNPq for the fellowship (Project 310785/2018-9). LAGS thanks CAPES for the doctoral scholarship (Grant 88887.639024/2014-01). 
	
	\bibliographystyle{bmc-mathphys}
	\bibliography{manuscript}     
\end{backmatter}

\end{document}